\documentclass[format=acmsmall]{acmart}
\usepackage{acm-ec-23-new}
\usepackage{booktabs} 
\usepackage[ruled]{algorithm2e} 

\SetAlFnt{\small}
\SetAlCapFnt{\small}
\SetAlCapNameFnt{\small}
\SetAlCapHSkip{0pt}
\IncMargin{-\parindent}

\setcitestyle{acmnumeric}

\makeatletter
\def\@ACM@checkaffil{
    \if@ACM@instpresent\else
    \ClassWarningNoLine{\@classname}{No institution present for an affiliation}%
    \fi
    \if@ACM@citypresent\else
    \ClassWarningNoLine{\@classname}{No city present for an affiliation}%
    \fi
    \if@ACM@countrypresent\else
        \ClassWarningNoLine{\@classname}{No country present for an affiliation}%
    \fi
}
\makeatother
    
\usepackage{pifont} 
\usepackage{xcolor} 
\usepackage{tikz} 
\usetikzlibrary{positioning}
\usetikzlibrary{arrows.meta,shapes,snakes}
\usepackage{enumitem} 
\usepackage{soul} 
\usepackage[normalem]{ulem}
\usepackage{subcaption} 

\usepackage{algorithmicx}
\usepackage{algpseudocode}
\algnewcommand\Input{\item[\textbf{Input:}]}
\algnewcommand\Output{\item[\textbf{Output:}]}

\newcommand{\N}{\mathbb{N}}

\newcommand{\Ch}{\mathcal{C}}
\newcommand{\q}{\mathbf{q}}
\newcommand{\ttup}[1]{\textup{\texttt{#1}}}
\newcommand{\noAA}{\ttup{BASE}}
\let\tilde\widetilde
\newcommand{\EOE}{\hfill $\triangle$} 

\usepackage{etoolbox}

\usepackage[compact]{titlesec}
\titlespacing{\section}{0pt}{2ex}{1ex}
\titlespacing{\subsection}{0pt}{1ex}{.5ex}
\titlespacing{\subsubsection}{10pt}{1ex}{.5ex}

\title[Disadvantaged Students in Highly Competitive Markets]{\Large Discovering Opportunities in New York City's Discovery Program: Disadvantaged Students in Highly Competitive Markets}

\author{Yuri Faenza} 
\affiliation{%
    \institution{IEOR, Columbia University}
    \city{New York, NY}
    \country{USA}
}
\author{Swati Gupta} 
\affiliation{%
    \institution{Sloan School of Management, MIT}
    \city{Cambridge, MA}
    \country{USA}
}
\authornote{Part of this work was done while the author was at Georgia Institute of Technology}
\author{Xuan Zhang} 
\affiliation{%
    \institution{IEOR, Columbia University}
    \city{New York, NY}
    \country{USA}
}

\begin{abstract}

        Discovery program (DISC) is a policy used by the New York City Department of Education (NYC DOE) to increase the number of admissions of students from low socio-economic background to specialized high schools. This policy has been instrumental in increasing the number of disadvantaged students attending these schools, by reserving a percentage of seats to disadvantaged students that complete a three-week summer program (with a very high success rate~\cite{NYT-HU}). However, assuming that students care more about the school they are assigned to rather than the type of seat they occupy (\emph{school-over-seat hypothesis}), our empirical analysis using NYC DOE data from 12 recent academic years (2005-06 to 2016-17) shows that DISC creates about 950 in-group blocking pairs each year amongst disadvantaged students, impacting about 650 disadvantaged students every year. Moreover, we find that this program does not respect improvements as it benefits lower-performing disadvantaged students more than top-performing disadvantaged students by matching some of the former to more preferred schools, thus unintentionally creating an incentive to under-perform. These experimental results are confirmed by our theoretical analysis.

    In order to alleviate the concerns caused by DISC, we explore two alternative policies: the minority reserve (MR) and the joint-seat allocation (JSA) mechanisms. As our main theoretical contribution, we introduce a feature of markets, that we term high competitiveness (HC).  Assuming the school-over-seat-hypothesis and the HC condition, we show that JSA dominates MR for all disadvantaged students. We give sufficient conditions under which high competitiveness is verified, such as the combination of high demand for seats and slightly poorer performances of disadvantaged students with respect to that of advantaged students. Data from NYC DOE satisfies the high competitiveness condition, and for this dataset our empirical results corroborate our theoretical predictions, showing the superiority of JSA. Given that JSA can be implemented by a simple modification of the classical deferred acceptance algorithm with responsive preference lists,we believe that, when the school-over-seat hypothesis holds, the discovery program can be changed for the better by implementing the JSA mechanism, leading in particular to aligned incentives for the top-performing disadvantaged students. We therefore suggest that policy makers solicit more information from students about the school-over-seat hypothesis and then explicitly incorporate their preferences in the mechanism.     
\end{abstract}

\begin{document}

\begin{titlepage}
    \maketitle
\end{titlepage}

\section{Introduction} \label{sec:introduction}

There is a pervasive problem in the way students are evaluated and given access to higher education \citep{ashkenas2017even,boschma2016concentration,capers2017implicit}. Promising students are often unable to join top schools because the path to getting admitted to these schools requires extensive training at various levels, starting as early as when students are $3$ years old \citep{shap-4yo}. Hence, underrepresented minorities, especially those with lower household income and lower family education, are systematically screened-out of the education pipeline: in many cities, schools remain highly segregated \citep{NYT1,NYT2}. Disparate opportunities in accessing high-quality education is one of the main causes of income imbalance and social immobility in the United States \citep{orfield2005segregation}. Policies such as quota-based mechanisms and training programs offer practical remedies for increasing representation of under-represented minorities and disadvantaged groups in public schools in the U.S. \citep{hafalir2013effective,dur2020explicit}, as well as in countries such as India \citep{sonmez2022affirmative} and Brazil \citep{aygun2021college}.

In this work, we study theoretically and empirically the characteristics of the \emph{Discovery Program}, which is used by the New York City Department of Education (NYC DOE) in an effort to increase the number of disadvantaged students at specialized high schools (SHS) \citep{SH-proposal}. SHSs span the five boroughs of NYC, and are among the most competitive ones in the city. Contrary to other public schools, these schools consider \emph{only} students' score on the Specialized High School Admissions Test (SHSAT) for admission. Around 5000 students are admitted every year to SHSs. The discovery program reserves some seats for disadvantaged students that are assigned after the regular admission process: it first runs the standard \emph{deferred acceptance} algorithm \cite{gale1962college} on general (i.e., non-reserved) seats with all student applicants, and it then runs the same algorithm on reserved seats for the unmatched disadvantaged students only. Disadvantaged students admitted via the reserved seats are required to participate in a 3-week enrichment program during the summer\footnote{The goal of this program is to better prepare disadvantaged students who are slightly below the cutoff points for attending specialized high schools. At the end of the program, each school decides whether to accept individual students based on, for example, their improvement. However, in practice, summer schools participation \emph{almost always} guarantees admission. E.g., in 2018, all students participating in the summer program at Stuyvesant were then admitted to the high school~\cite{NYT-HU}.}.

The discovery program has been instrumental in creating opportunities for disadvantaged students, increasing the number of admitted students to these extremely competitive public high schools in NYC. In 2020, for example, Mayor Bill de Blasio called for an expansion of the discovery program, with 20\% seats at SHSs reserved for the program. This expansion resulted in $1,350$ more disadvantaged students being admitted to these specialized schools \citep{SH-proposal,chalkbeat}. 

In this work, we dive deep into the student-school matching produced by the discovery program and starting from this case study, deduce general properties of markets with similar features. Our empirical analysis shows that under a reasonable assumption on students' preferences over schools which we term \emph{school-over-seat}\footnote{This hypothesis assumes that students' preference over schools are not affected by whether they are required to participate in the three-week summer enrichment program. See Section \ref{sec:policy-recommendation} for further discussions.}, the matchings from academic years 2005-06 to 2016-17 created about $950$ in-group blocking pairs each year amongst disadvantaged students, impacting about $650$ disadvantaged students every year (see Figure \ref{fig:block-pair-disc}). A blocking pair is a pair of student $s_1$ and school $c_1$ that prefer each other to their matches, thus violating the priority of student $s_1$ at school $c_1$ and creating dissatisfaction among students and schools. We also find that this program does not respect improvements, hence it benefits {\it lower-performing} disadvantaged students more than {\it top-performing} disadvantaged students, thus unintentionally creating an incentive to under-perform. See Figure \ref{fig:disc-who-worse} for our empirical analysis, where top-performing students (with ranks $0\sim 500$) attend less preferred schools  under the discovery program, unlike low-performing students (with ranks $500\sim 2200$) who get matched to better ranked schools (more preferred schools have lower numeric ranks). These drawbacks\footnote{These drawbacks are also discussed in online forums: see, e.g., a post on the r/SHSAT subreddit: \url{https://www.reddit.com/r/SHSAT/comments/ntkoq5/discovery_program/}; and a discussion on a popular site referenced on many reddit posts: \url{https://www.gregstutoringnyc.com/shsat-Discovery/}.} are not just an artifact of the data from NYC DOE, but are theoretical problems with the current implementation of the discovery program and the nature of the market. 

\begin{figure}[t]
\centering
\begin{subfigure}[t]{.48\textwidth}
\centering
\includegraphics[width=\textwidth]{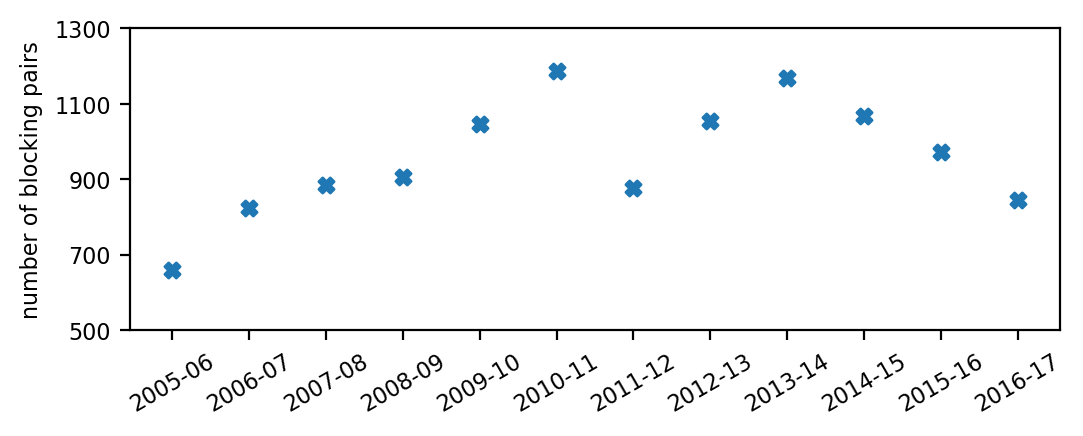}
\caption{\raggedright Number of blocking pairs among disadvantaged students under discovery program across the last 12 academic years, which impacted around 650 students each year. \vspace{-1em}} \label{fig:block-pair-disc}
\end{subfigure}
\hspace{.02\textwidth}
\begin{subfigure}[t]{.48\textwidth}
    \centering
    \includegraphics[width=\textwidth]{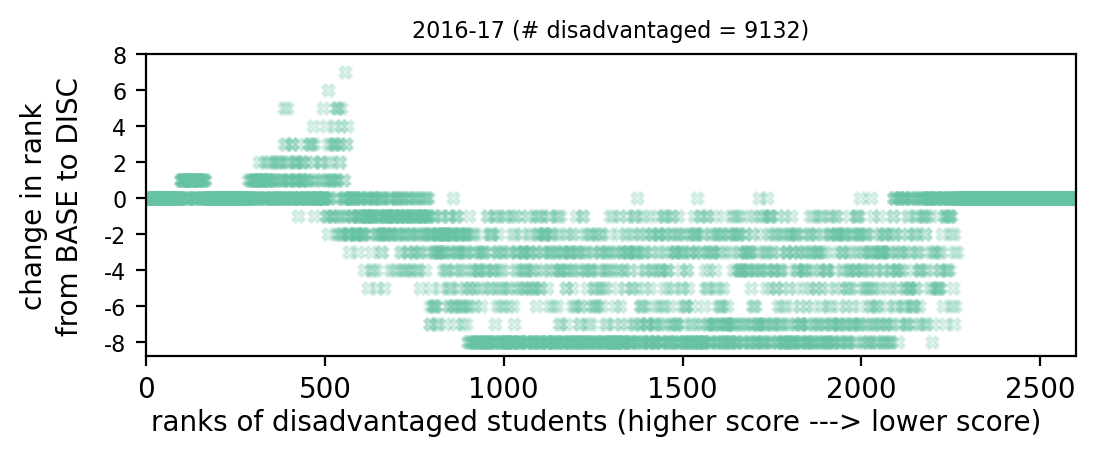}
    \caption{\raggedright Change in rank of assigned schools from the baseline mechanism (\noAA) to discovery program (\ttup{DISC}) (we plot \ttup{DISC} - \noAA) for disadvantaged students, ordered by their SHSAT scores.  \vspace{-1em}}
    \label{fig:disc-who-worse}
\end{subfigure}
\caption{Note that for Figure \ref{fig:disc-who-worse}, a negative change in rank means getting to a more preferred school under the discovery program. Top-performing disadvantaged students (ranked $0\sim 500$) are matched to worse schools under DISC, whereas the lower-performing disadvantaged students are matched to better schools.}
\end{figure}

Therefore, our goal in this paper is to explore alternative mechanisms with reserved seats, so that we can propose practical modifications to how the discovery program is implemented, while alleviating the above-mentioned drawbacks. We want to propose a solution that is theoretically sound, benefits disadvantaged students the most, and applies more generally to markets similar to the NYC SHS market under consideration, where demand vastly exceeds the offer. 

We consider in particular the \emph{minority reserve} (\ttup{MR}) and \emph{joint seat allocation} (\ttup{JSA}) mechanisms\footnote{They differ in terms of the processing order between reserved and general seats, which is known to affect the resulting matching. \ttup{MR} and \ttup{JSA} are referred to as \emph{horizontal} and \emph{vertical} reservation respectively in some of the literature.} These two mechanisms are also quota-based, with schools reserving a certain proportion of their seats for disadvantaged students. However, in contrast to the discovery program, \ttup{MR} and \ttup{JSA} allocate reserved and general seats at the same time. Under minority reserve, disadvantaged students are admitted first via reserved seats and then via general seats (when there are no more reserved seats); whereas under joint seat allocation, disadvantaged students would take first general seats if they are able to compete and otherwise revert to reserved seats. We compare these three policies with respect to the baseline stable matching mechanism, \noAA, which does not distinguish between disadvantaged and advantaged students~\citep{gale1962college}. We next discuss our key contributions. 

\subsection{Main results} \label{sec:main-results}

We first summarize known and new properties of different mechanisms with reserved seats under the school-over-seat hypothesis, i.e., students' preferences over schools are not influenced by whether they are admitted via general seats or reserved seats (in the case of NYC SHSs, reserved seats additionally require a 3-week summer program). We then show that while in general no mechanism dominates the other (not even \ttup{DISC}), we prove as our main theoretical contribution that \texttt{JSA} dominates \ttup{MR} for disadvantaged students under a novel, fairly broad condition, that we term \emph{high competitiveness} of markets. Finally, we empirically validate the high competitiveness condition and our theoretical results using data from NYC DOE, and make a policy recommendation for the discovery program. 

\subsubsection{Properties of Mechanisms.}

\paragraph{Question 1.} \textit{Which mechanisms with reserved seats considered in the paper satisfy reasonable notions of fairness such as absence of in-group blocking pairs and strategy-proofness? What is the impact of these mechanisms on the disadvantaged group of students?} 
\smallskip

We explore useful properties for mechanisms with reserved seats and briefly explain these properties here (see Sections \ref{sec:model-notation} and \ref{sec:affirmative-action} for formal definitions): 

\begin{itemize}
    \item[(i)] {\it strategy-proofness}: this property means that the best strategy of students is to honestly report their preferences over the schools; 
    \item[(ii)] {\it absence of in-group blocking pairs}: this is a \emph{fairness} condition which ensures there is no priority violation for students;
    \item[(iii)] the third property asks for the mechanism not to worsen (with respect to the baseline mechanism \noAA) the assignment of \emph{at least} one disadvantaged student\footnote{Intuitively, one might expect property (iii) to be so weak that it is trivially satisfied. However, the discovery program does not satisfy it in general.}; 
    \item[(iv)] the fourth property asks \emph{all} disadvantaged students not to be worse-off in a restricted scenario called \emph{smart-reserve}\footnote{This requirement was first proposed and studied by \citet{hafalir2013effective}, and they showed that such a condition is achievable either in an ad-hoc fashion or by using historical data on school admissions.}, meaning the number of seats reserved for disadvantaged students is no less than the number of disadvantaged students admitted by the baseline mechanism;
    \item[(v)] {\it respect for improvements}: this property is essential for meritocratic systems and it ensures that students have no incentive to underperform in the exam (i.e., lower their priority standings).
\end{itemize}

\begin{small}
\begin{table}[t]
    \centering
    \begin{tabular}{l|cccccccc}
        & \noAA & \ttup{DISC} & \ttup{MR} & \ttup{JSA} \\
        \hline
        weakly group strategy-proof & \ding{51} {\small [DF]} & \ding{55} (Ex \ref{ex:disc-fair-sp}) & \ding{51} {\small [HYY]} & \ding{51} ([KS,AT]) \\
        no in-group blocking pairs & \ding{51} {\small [GS]} & \ding{55} (Ex \ref{ex:disc-fair-sp}) & \ding{51} (Prop \ref{prop:mr-no-bp}) & \ding{51} (Prop \ref{prop:jsa-no-bp}) \\
        at least one disadvantaged student not worse off & \ttup{NA} & \ding{55} (Ex \ref{ex:disc-all-worse}) & \ding{51} {\small [HYY]} & \ding{51} (Thm \ref{thm:jsa-one-better}) \\
        no disadvantaged student worse off if smart reserve & \ttup{NA} & \ding{55} (Ex \ref{ex:disc-fair-sp}) & \ding{51} {\small [HYY]} & \ding{51} (Thm \ref{thm:jsa-smart-all-better}) \\
        {respect for improvements} & \ding{51}[KS] & \ding{55} (Ex \ref{ex:disc-fair-sp}) & \ding{51} [KS,AT] & \ding{51} [KS,AT] \\ \hline
    \end{tabular}
    \vspace{0.1cm}
    \caption{\raggedright Summary of properties of mechanisms with reserved seats under the school-over-seat assumption. \ttup{NA} means not applicable. Previously known results and their corresponding citations are given in square brackets, with: [DF] \citet{dubins1981machiavelli}; [HYY] \citet{hafalir2013effective}; [GS] \citet{gale1962college}; [KS] \cite{kominers2016matching}; and [AT] \citet{aygun2020dynamic}; other results are accompanied by the labels of examples, propositions, or theorems used to answer the questions. \vspace{-2.5em}} 
    \label{tab:summary-prop}
\end{table}
\end{small}

We summarize our results, as well as known results from the literature, in Table \ref{tab:summary-prop}. As one can immediately see from the table, the current implementation of the discovery program does not satisfy any of the attractive features we investigate, yet the other two mechanisms, \texttt{MR} and \texttt{JSA}, satisfy all these properties. This is even true when {\it all the schools rank students in the same order}, as in the NYC SHS admission market where students are ranked based on their SHSAT scores. We additionally demonstrate these findings empirically by computational experiments using the admission data on NYC SHSs (the details can be found in Section \ref{sec:data}). These results suggest that the discovery program could benefit by replacing the current implementation with either minority reserve or joint seat allocation. This result calls for a direct comparison of those mechanisms. 

\subsubsection{Dominance across Mechanisms with Reserved Seats}

\paragraph{Question 2.} \textit{Considering a fixed reservation quota. Does one of the mechanisms with reserved seats (\ttup{DISC}, \ttup{JSA} or \ttup{MR}) \emph{(weakly) dominate} another one for disadvantaged students, i.e., do all disadvantaged students weakly prefer the schools they are matched to under one mechanism compared to the other?}
\smallskip

We say that a mechanism A \emph{(weakly) dominates} another mechanism B for disadvantaged students if A places all disadvantaged students in schools they like at least as much as the schools they are placed in by B. Our results from Table \ref{tab:summary-prop} seem to suggest that the discovery program mechanism could be dominated by either minority reserve or joint seat allocation. However, this is \emph{not} the case, as shown by the results we summarize in Table \ref{tab:summary-compare}. All three mechanisms are incomparable, even under some pretty restrictive hypothesis: (1) schools rank students in the same order; and/or (2) reservation quotas being a smart reserve. The first hypothesis is common in markets where students' ranking is based on an entrance exam, such as the one for NYC SHSs, Chinese universities, and Indian IITs. The only exception to the incomparability results is that the baseline mechanism \noAA, under the second hypothesis, is dominated by minority reserve and joint seat allocation\footnote{This exception is simply another way of expressing the same results related to the third property in Table \ref{tab:summary-prop}.}.

\begin{small}
\begin{table}[t]
    \centering
    \setlength{\tabcolsep}{1pt}
    \begin{tabular}{l|lll|lll|lll|lll}
        \hline
        & \multicolumn{3}{c|}{\noAA} & \multicolumn{3}{c|}{\ttup{MR}} & \multicolumn{3}{c|}{\ttup{DISC}} & \multicolumn{3}{c}{\ttup{JSA}} \\
        \hline
        \noAA \hspace{2pt}
        & & & 
        & (\ding{55}) & (\ding{55}) & (\ding{55} Ex \ref{ex:AA-better}) 
        & (\ding{55}) & (\ding{55}) & (\ding{55} Ex \ref{ex:AA-better}) 
        & (\ding{55}) & (\ding{55}) & (\ding{55} Ex \ref{ex:AA-better}) \\
        \ttup{MR} 
        & (\ding{55} {\small [HYY]}) & (\ding{51} {\small [HYY]}) & (\ding{51}) 
        & & & 
        & (\ding{55}) & (\ding{55}) & (\ding{55} Ex \ref{ex:disc-no-compare}) 
        & (\ding{55}) & (\ding{55}) & (\ding{55} Ex \ref{ex:mr-jsa-no-compare}) \\
        \ttup{DISC} 
        & (\ding{55}) & (\ding{55}) & (\ding{55} Ex \ref{ex:disc-fair-sp}) 
        & (\ding{55}) & (\ding{55}) & (\ding{55} Ex \ref{ex:disc-no-compare}) 
        & & & 
        & (\ding{55}) & (\ding{55}) & (\ding{55} Ex \ref{ex:disc-no-compare}) \\
        \ttup{JSA} 
        & (\ding{55} \ref{ex:jsa-mr-worse}) & (\ding{51} Thm \ref{thm:jsa-smart-all-better}) & (\ding{51}) 
        & (\ding{55}) & (\ding{55}) & (\ding{55} Ex \ref{ex:mr-jsa-no-compare}) 
        & (\ding{55}) & (\ding{55}) & (\ding{55} Ex \ref{ex:disc-no-compare}) \\\hline 
    \end{tabular}
    \vspace{0.1cm}
    \captionsetup{justification=raggedleft,singlelinecheck=false}
    \caption{\raggedright The table answer the following question under the school-over-seat assumption: does the ``row'' mechanism dominates the ``column'' mechanism for disadvantaged students? We answer the question for three restricted domains: (1) schools share a common ranking of the students, (2) the reservation quotas is a smart reserve, and (3) both. The answers are given in the exact order. All answers are accompanied by the citations with [HYY] \citet{hafalir2013effective} or the labels of the examples or theorems used to answer the questions, except for cases when the answer for one domain can be inferred from that of another domain. \vspace{-2.5em}} 
    \label{tab:summary-compare}
\end{table} 
\end{small}

\paragraph{Question 3.} \textit{Between \ttup{MR} and \ttup{JSA}, is one better for the NYC SHS market under consideration? Can we deduce general properties that imply domination between these two mechanisms?}
\smallskip

To be able to identify crucial interventions for the discovery program, we study the behavior of the \ttup{JSA} and \ttup{MR} mechanisms in markets that satisfy a condition which we call \emph{high competitiveness}. 
This is a novel ex-post condition which guarantees that \texttt{JSA} weakly dominates \texttt{MR} for disadvantaged students. This condition is verified by our data from NYC DOE, where in fact \texttt{JSA} outperforms \texttt{MR} for disadvantaged students. We also show reasonable conditions on the primitives of the market that imply high competitiveness.  See Theorems \ref{thm:jsa-dom-mr-condition}, \ref{thm:cs-highly-competitive}, and \ref{thm:hc-potential} for the formal statement. Roughly speaking, the high competitiveness condition is satisfied when the demand for seats (i.e., number of students) is much larger than the supply, and when disadvantaged students are performing systematically worse than advantaged students. For the latter, we compare the distribution of SHSAT scores for both the advantaged and disadvantaged groups of students, and notice that there is a distributional shift between the scores of these two groups of students (see Figure \ref{fig:score-distributions}).

\subsubsection{Case Study based on Data from New York City's Department of Education.}\label{sec:subsub:intro-case-study}

We validate our theoretical results with extensive computational experiments using data we obtained from NYC DOE for the 2005-2006 to 2016-2017 academic years, where we label students as advantaged or disadvantaged based on the criteria given by the discovery program. First, we show that, in practice as well, the discovery program suffers from many of the theoretical drawbacks we presented in Table~\ref{tab:summary-prop} -- in particular, the discovery program creates in-group blocking pairs (Figure \ref{fig:block-pair-disc}) and does not respect improvements (Figure \ref{fig:disc-who-worse}). In terms of strategy-proofness, we are unable to observe systematic strategic behaviors from disadvantaged students. However, this is not surprising, as it is hard to detect if a preference list has been manipulated and it has moreover been well observed in the literature that strategic behaviors are unlikely to occur in large markets due to lack of information (see, e.g., \cite{kesten2010school,roth1999truncation,kojima2009incentives}). See section \ref{sec:res-dis} for details.

When comparing the mechanisms with the same reservation quotas, we observe that the discovery program results in the highest number of disadvantaged students admitted, whereas minority reserve has the lowest amount (see Figure \ref{fig:perc_in_school}). One may be tempted to deduce that the discovery program is the best for the disadvantaged group of students as a whole. However, the number of admits can easily be increased at the policy-maker's will by increasing the number of the reserved seats, while the negative impact of unfair seat allocation cannot be dealt with by a simple perturbation of the parameters.

In addition, given the observation that disadvantaged students are in general performing worse than advantaged students (see Figure \ref{fig:score-distributions}), it would undoubtedly lead to underrepresentation of disadvantaged students at these SHSs under the baseline mechanism (see Figure \ref{fig:perc_in_school}). Together with the fact that there is a limited number of seats when compared to the number of students applying to SHSs, we expect the market to be highly competitive and thus all disadvantaged students would weakly prefer their assignment under \ttup{JSA} than under \ttup{MR}. We indeed observe these characteristics for the NYC SHS admission market across all academic years we have data for (see Figure~\ref{fig:mr-jsa-adv-dis}). This leads to the policy recommendation we present in this work.

\subsubsection{Policy Recommendation.} \label{sec:policy-recommendation} 

Overall, our work paves the way to make the discovery program fairer for disadvantaged students. In particular, we provide an answer to how the existing practice of the discovery program can be changed minimally to improve the outcome for the disadvantaged group of students, so that the program aligns with the incentives to perform better.

\paragraph{Our Proposal:} We propose that the program takes into account the preferences of students in terms of the schools versus seats. Is attending a particular school more important than the type of seat they are assigned to or vice versa? We believe that most students should be willing to take a one-time 3-week summer program to attend a school they prefer, rather than not taking the program and attending, for 4 years, a school they prefer less. We find that this school-over-seat hypothesis is supported by the fact that preferences appear to be strongly polarized for certain schools due to, e.g., geographical considerations (details are reported in the Appendix, Section \ref{sec:res-preference}). Although this seems reasonable, unfortunately such preferences\footnote{In case where the school-over-seat assumption does not hold for a significant amount of students, a \emph{direct} mechanism which explicitly asks students to rank \emph{contracts} at schools (i.e., school and seat type) would be more suitable. \citet{sonmez2013matching} developed a mechanism that is fair, strategy-proof, and respects improvements, and the mechanism can be easily implemented for the NYC SHS market: order the students based on their SHSAT scores with ties broken randomly, then one at a time assign to students their most preferred contract that are still available.}  are currently not collected in the data provided by the NYC DOE.

Under the school-over-seat assumption, we find that the many drawbacks of the current implementation of the discovery program can be corrected by following the {\it joint seat allocation} mechanism. For the NYC Specialized High School market -- and, more generally, for highly competitive markets -- joint seat allocation gives a matching that is weakly better for disadvantaged students, when compared to matching output by the other replacement mechanism studied in this paper, both in theory and in practice. 

Although powerful, the modification we propose requires \emph{minimal} modification: there is essentially no change in terms of what students and schools should report to the DOE (preference lists for both and admission capacity for schools), and there is no change in terms of the algorithm (the deferred acceptance algorithm \citep{gale1962college}, which is currently in implementation). Given this information, to implement the \ttup{JSA} mechanism, one only needs to compute an equivalent instance \citep{kominers2016matching,aygun2020dynamic} where students' preference lists are expanded to be over reserved and general seats at schools, so that the matching we desire to obtain can be easily recovered from the matching obtained under the classical stable matching model on this equivalent instance.

\medskip
Details of the equivalent instance can be found in Appendix \ref{sec:aux-instances}. Before we delve deeper into our model and results, we would like to highlight a trade-off that any constrained resource allocation problem faces. Diverting resources to the disadvantaged groups can result in taking some resources that are currently assigned to the advantaged groups. In this work as well, we find from our empirical analysis, that advantaged students always weakly prefer their assignment under \ttup{MR} compared to \ttup{JSA}. For all the academic years we analyze, we find that about $3\%$ of the advantaged students are worse off under \ttup{JSA} than under \ttup{MR} (i.e., about $97\%$ of them are matched to the same school under the two mechanisms); and among the $3\%$, most of them experience a drop in the rank of assigned schools that is at most two. See Figure \ref{fig:mr-jsa-adv-dis} for details of one academic year. We consider this impact to be minimal compared to the ill-treatment faced by the disadvantaged students.

\subsection{Related literature} \label{sec:literature}

The problem of assigning students to schools (without reserving seats for disadvantaged students) was first studied by Gale and Shapley in their seminal work \citep{gale1962college}. \citet{abdulkadirouglu2003school} then analyzed the algorithm in the context of school choice and recommended school districts to replace their current mechanisms with either this algorithm or another algorithm, called the \emph{top trading cycle algorithm}. Since then, these mechanisms have been widely adopted by many cities in the United States, such as New York City and Boston. 

The first attempt of incorporating seat reservation with the stable mechanism occurred in this pioneering work \citep{abdulkadirouglu2003school}, where they extended their analysis to a simple affirmative action policy, using \emph{majority quotas}. However, \citet{kojima2012school} then analyzed the effects of these proposed affirmative action policies, as well as priority-based policies, and showed that in some cases, the mechanisms might hurt disadvantaged students, the very group these policies are trying to help. \citet{hafalir2013effective} further analyze the effect empirically through simulated data and suggested that this phenomenon might be quite common, and does not just happen in theory due to special edge cases. In addition, to overcome the efficiency loss, they propose the minority reserve mechanism.

Since then, there has been many work studying and proposing solutions for the efficiency loss due to seat reservation, such as \citet{afacan2016affirmative,dougan2016responsive,echenique2015control,ehlers2014school,fragiadakis2017improving,jiao2021school,nguyen2019stable}. 

Mostly related to our work are those that study the effects of the precedence order under which different types of seats are allocated (a special case is when there are only two types: reserved and general). \citet{kominers2016matching} is the first to study the importance of this precedence order in an abstract and general framework. \citet{dur2018reserve} then extended upon this work in the context of school choice and show its role in explaining why the walk zone reserve in Boston does not have the intended impact. Motivated by a school choice application in Chicago, \citet{dur2020explicit} compare mechanisms where multiple tiers of students are present, and each tier have some seats where they have priority over other tiers. Such mechanisms include, in particular, \ttup{JSA} and \ttup{MR}. They show that the precedence order can provide an ``additional lever to explicitly target disadvantaged applicants". In particular, their main results imply conditions under which the \emph{number} of disadvantaged students admitted by \ttup{JSA} is at least that by \ttup{MR}. In comparison, we give conditions on the \emph{quality} of the matching \emph{for each disadvantaged student}. As discussed in Section~\ref{sec:subsub:intro-case-study}, we believe that comparing mechanisms in terms of the quality of the matching for individual disadvantaged students gives a perspective complementary to the one that looks at the number of admitted disadvantaged students. Moreover, the hypotheses of~\cite{dur2020explicit} on preferences of agents appear to be more restrictive than ours; for instance, they require all schools to have the same ranking of the students. 

Other works that study precedence order in school choice include \citet{sonmez2022affirmative} for India's affirmative action system with both vertical and horizontal reservation policies and \citet{aygun2021college} for Brazil's affirmative action system. Moreover, \citet{delacretaz2021processing} proposed a simultaneous reserve system that treats all types of seats identically. \citet{pathak2023reversing} studied precedence order from a policy perspective, showing how misunderstanding of the precedence order or the reserve system affects decisions from applications of reserve systems. There are also works that studied other real-world applications besides school choice, such as H1B-visa allocation \citep{pathak2022immigration}, vaccine allocation \citep{pathak2021fair}, and cadet-branch matching in U.S.~military \citep{sonmez2013matching}.

Another popular form of affirmative action are \emph{priority-based} mechanisms (see, e.g., \cite{hafalir2013effective,jiao2021school,kojima2012school}), which creates a higher priority for disadvantaged students by, e.g., boosting their scores. Though this mechanism  satisfies important properties such as strategy-proofness and absence of in-group blocking pairs, its practical use is being largely debated. For example, in 2019, the college board proposed adding an adversity score to SAT scores to account for socio-economic differences, however, this was met with severe pushback \citep{sat-adversity}. In another lawsuit at the University of Michigan challenging a priority-based mechanism that assigned 20 points extra to disadvantaged students, the system was declared unconstitutional by the Supreme Court \citep{Gratz.vs.Bollinger}.~\cite{faenza2020impact} investigates the effects of policies where scores for minority students are boosted before the admission process by extra training, additional resources, etc. Since the goal of this work is to focus on operational suggestions to the discovery program, we do not explore priority-based mechanisms. 

\subsection{Outline}

The rest of the paper is organized as follows. In Section \ref{sec:model-notation}, we introduce the basic model and related concepts for stable matchings and stable matching mechanisms. In Section \ref{sec:affirmative-action}, we formally introduce the mechanisms with reserved seats considered in this paper and investigate their properties and answer Question 1. We then compare these mechanisms in Section \ref{sec:aa-compare} and provide the answer to Question 2 and Question 3. Lastly, in Section \ref{sec:data}, we dive into the data on NYC SHS admission, demonstrate our theoretical findings empirically and provide additional observations.

\section{Model and Notations} \label{sec:model-notation}

\subsection{Matchings and mechanisms}

Let $S$ and $C$ denote a finite set of students and schools, respectively. Let $G=(S\cup C, E)$ be a bipartite graph, where two sides of nodes are students and schools, and the edge set $E$ represents the schools which students find \emph{acceptable} (i.e., would like to attend). Every student $s\in S$ has a strict preference relation $>_s$ (which we call the \emph{preference list} of student $s$) over the schools they find acceptable and the option of being unassigned (denoted by $\emptyset$). Formally, for two options $c_1, c_2\in C\cup \{\emptyset\}$, $c_1 >_s c_2$ means that student $s$ strictly prefers $c_1$ to $c_2$. For every student-school pair $(s,c)$, we let $c>_s \emptyset$ if $(s,c)\in E$, and $\emptyset>_s c$ otherwise. There are two types of students, \emph{advantaged} (or \emph{majority}) and \emph{disadvantaged} (or \emph{minority}), denote by $S^M$ and $S^m$ respectively. That is, $S=S^M \dot\cup S^m$ where $\dot\cup$ is the disjoint union operator. On the other hand, every school $c$ has a quota $q_c\in \N\cup \{0\}$, which represents the maximum number of students it can admit, and a strict \emph{priority order} $>_c$ over the students: for any two students $s_1, s_2\in S$, $s_1 >_c s_2$ means that student $s_1$ has a higher \emph{priority} (e.g., higher test score) than student $s_2$ at school $c$. 

Let $>_S\equiv \{>_s: s\in S\}$, $>_C\equiv \{>_c: c\in C\}$, and $\q \equiv \{q_c: c\in C\}$ denote the collection of students' preference lists, their priority orders at schools, and schools' quotas, respectively. Moreover, we write $>\equiv \{>_S, >_C\}$. An \emph{instance (or market)} is thus denoted by $(G, >_S, >_C, \q)$ or $(G, >, \q)$. 

A \emph{matching} $\mu$ (of an instance) is a collection of student-school pairs such that every student is incident to at most one edge in $\mu$ and every school $c$ is incident to at most $q_c$ edges in $\mu$. For student $s\in S$ and school $c\in C$, we denote by $\mu(s)$ the school student $s$ is matched (or assigned) to, and by $\mu(c)$ the set of students school $c$ is matched (or assigned) to, under matching $\mu$.

For every school $c\in C$, let $q_c^R\in \{0,1,\cdots, q_c\}$ denote the number of seats reserved to disadvantaged students at school $c$, and let $q_c^G \coloneqq q_c-q_c^R$ denote the number of general seats at school $c$. We call $\q^R \coloneqq \{q^R_c: c\in C\}$ the \emph{reservation quotas}. A \emph{(matching) mechanism with reserved seats} is a function that maps every instance, together with reservation quotas, to a \emph{matching}. Given an instance $I= (G, >, \q)$, a mechanism $\phi$, and reservation quotas $\q^R$, let $\phi(I, \q^R)$ denote the matching obtained under the mechanism $\phi$ with reservation quotas $\q^R$. Sometimes, when the reservation quotas are clear from context, we simply denote the matching as $\phi(I)$. 

We say priority orders $\tilde{>}_C$ is an \emph{improvement} of $>_C$ for student pair $s\in S$ if $\tilde{>}_C$ is obtained from $>_C$ by increasing the priorities of student $s$ in some schools in $C$, while leaving the relative priority orders of other students unchanged. A mechanism is said to \emph{respect improvements} if for any instance $I=(G,>_S,>_C,\q)$, student $s\in S$, and an improvement $\tilde >_C$ of $>_C$ for student $s$, we have that $\phi(\tilde I, \q^R)(s) \ge_s \phi(I, \q^R)(s)$, where $\tilde I$ is obtained from $I$ by replacing $>_C$ with $\tilde >_C$.

Let $\mu_1, \mu_2$ be two matchings. We say $\mu_1$ \emph{(weakly) dominates} $\mu_2$ \emph{for disadvantaged students} if $\mu_1(s) \ge_s \mu_2(s)$ for all disadvantaged students $s\in S^m$. If moreover $\mu_1\neq \mu_2$ (i.e., there is at least one disadvantaged student $s\in S^m$ such that $\mu_1(s) >_s \mu_2(s)$), then we say $\mu_1$ \emph{Pareto dominates} $\mu_2$ \emph{for disadvantaged students}. A student-school pair $(s,c)\in E$ is a blocking pair of matching $\mu$ \emph{for disadvantaged students} if $s\in S^m$, $c>_s \mu(s)$, and there exists a disadvantaged student $s'\in \mu(c)\cap S^m$ such that $s>_c s'$; and it is a blocking pair of matching $\mu$ \emph{for advantaged students} if $s\in S^M$, $c>_s \mu(s)$, and there exists an advantaged student $s'\in \mu(c)\cap S^M$ such that $s>_c s'$. A blocking pair is called an \emph{in-group} blocking pair if it is a blocking pair for either disadvantaged or advantaged students. 

Fix reservation quotas $\q^R$. A mechanism $\phi$ is \emph{strategy-proof} if for any instance $I$ and for any student $s\in S$, there is no preference list $\tilde >_s$ such that $\phi(\tilde I, \q^R)(s) >_s \phi(I, \q^R)(s)$, where $\tilde I$ is obtained from $I$ by replacing $>_s$ with $\tilde >_s$. In other words, a mechanism is strategy-proof if no student has the incentive to misreport their preference list. As a stronger concept, a mechanism is \emph{weakly group strategy-proof} if for any instance $I$ and for any \emph{group} of students $S_1 \subseteq S$, there are no preference lists $\{\tilde >_s: s\in S_1\}$ such that for every student $s\in S_1$, $\phi(\tilde I, \q^R)(s) >_s \phi(I, \q^R)(s)$, where $\tilde I$ is obtained from $I$ by replacing $>_s$ with $\tilde >_s$ for every $s\in S_1$. That is, a mechanism is weakly group strategy-proof if no group of students can jointly misreport their preference lists so that everyone in the group is strictly better off. Note that if a mechanism is weakly group strategy-proof, it is strategy-proof.

Consider two mechanisms $\phi_1$ and $\phi_2$. If $\phi_1(I, \q^R)$ (weakly) dominates $\phi_2(I, \q^R)$ for disadvantaged students for all instances $I$, we say that mechanism $\phi_1$ \emph{(weakly) dominates} mechanism $\phi_2$ for disadvantaged students. If neither $\phi_1$ nor $\phi_2$ dominates the other mechanism, we say they are \emph{not comparable} or \emph{incomparable}.

\subsection{Choice functions} \label{sec:choice-function}

To unify the treatment of the different mechanisms seen in the paper, we next introduce the concept of choice functions. Under each mechanism, every school $c\in C$ is endowed with a \emph{choice function} $\Ch_c: 2^S \rightarrow 2^S$: for every subset of students $S_1\subseteq S$, $\Ch_c(S_1)$ represents the students whom school $c$ would like to admit among those in $S_1$. In particular, for every $S_1\subseteq S$, we have $\Ch_c(S_1) \subseteq S_1$ and $|\Ch_c(S_1)|\le q_c$. Choice function $\Ch_c$ is a function of the priority order $>_c$ and quotas $q_c^R$ and $q_c^G$, and its exact definition depends on the specific mechanism (see Section \ref{sec:aa-compare}). Students' preferences are still described by a strict order over a subset of schools.

For all mechanisms studied in this paper, every school $c$'s choice function $\Ch_c$ satisfies the following classical (see, e.g.,~\cite{alkan2002class}) properties: \emph{substitutability}, \emph{consistency}, and \emph{$q_c$-acceptance}\footnote{$q_c$-acceptance is also referred to as \emph{quota-filling} by some authors. However, we prefer to use $q_c$-acceptance since it highlights the quota.}. Thus, for the rest of the paper, unless otherwise specified, these properties are always assumed. For some mechanisms, $\Ch_c$ is additionally \emph{$q_c$-responsive}. Intuitively, Substitutability states that whenever a student is selected from a pool of candidates, they will also be selected from a smaller subset of the candidates; consistency is also called ``irrelevance of rejected contracts'', which means that removing rejected candidates from the input does not change the output; $q_c$-acceptance means that the choice function fills the $q_c$ positions as much as possible; and $q_c$-responsiveness means that there is an underlying priority order over the students and the choice function simply selects $q_c$ students with the highest priorities whenever available. Formal definitions are included in Appendix \ref{sec:def-ch-func}.

For any nonnegative integer $q$, a priority order over the students $>$, and a subset of students $S_1\subseteq S$, let $\max(S_1, >, q)$ denote the $\min(q, |S_1|)$ highest ranked students (i.e., students with the highest priorities) of $S_1$ according to the priority order $>$. We further note that $q$-responsiveness implies substitutability, consistency, and $q$-acceptance. Indeed, $q$-responsive choice functions are the ``simplest'' choice functions and are mostly studied in the matching literature, including the seminal work by~\citet{gale1962college} and in practical school choice~\citep{NYC,BOSTON}.

\subsection{Stable matchings} 

Consider an arbitrary collection of schools' choice functions $\Ch\coloneqq \{\Ch_c: c\in C\}$. Note that the $q_c$-acceptant property implies that for every school $c$, we must have $\Ch_c(\mu(c)) = \mu(c)$ by any matching $\mu$ by the definition of matchings. A matching $\mu$ is \emph{stable} (in instance $I$ under choice functions $\Ch$) if there is no student-school pair $(s,c)\in E$ such that $c>_s \mu(s)$ and $s\in \Ch_c(\mu(c)\cup \{s\})$. When such a student-school pair exists, we call it a \emph{blocking pair} of $\mu$, or we say that the edge (or pair) \emph{blocks} $\mu$. Note that the definition of matchings only depends on the instance, not on the choice functions; whereas the definition of stability depends on both.

When the choice function is $q_c$-responsive (i.e., induced by a priority order and a quota), the definition of stability with respect to choice functions is equivalent to the standard definition in the classical model without choice functions. In particular, the condition $s\in \Ch_c(\mu(c)\cup \{s\})$ can then be stated as: either school $c$'s seats are not fully assigned (i.e., $|\mu(c)|<q_c$) or $s$ has a higher priority over some students that are assigned to $c$ (i.e., $\exists s'\in \mu(c)$ such that $s>_c s'$).

Among all stable matchings of a given instance and choice functions, there is one that \emph{dominates} every stable matching, where matching $\mu_1$ is said to \emph{dominate} matching $\mu_2$ if $\mu_1(s) \ge_s \mu_2(s)$ for all students $s\in S$.  This stable matching is called the \emph{student-optimal} stable matching, and it can be obtained by the \emph{student-proposing deferred acceptance algorithm}~\citep{gale1962college,roth1984stability}, which we describe next. The algorithm runs in \emph{rounds}. At each round $k$, every student applies to their most preferred school that has not rejected them; and every school $c$, with $S^{(k)}_c$ denoting the set of students who applied to it in the current round, \emph{temporarily} accepts students in $\Ch_c(S^{(k)}_c)$ and rejects the rest. The algorithm terminates at the first iteration $k$ when there is no rejection and outputs the matching $\mu$ with $\mu(c)=S_c^{(k)}$ for every school $c$. For any instance $I$ and choice functions $\Ch$, we denote by $\ttup{SDA}(I, \Ch)$ the matching output by the student-proposing deferred acceptance algorithm.

\section{Mechanisms} \label{sec:affirmative-action}

For the rest of the section, we fix an instance $I=(G, >, \q)$ and reservation quotas $\q^R$. The choice functions of schools depend on the mechanisms, and we introduce them in details in each subsection. We also discuss the features of the mechanisms in their corresponding subsections.  The unified treatment presented here allow us to compare the different mechanisms, with the goal of understanding which is the one that best fit our improvement goals for the discovery program. We defer all proofs in this subsection to Appendix \ref{sec:app:missing-proof}.

\subsection{Baseline mechanism} \label{sec:noAA}

The simplest mechanism is the one where schools do not distinguish students of different types. In this case, the choice function of school $c$ under this baseline mechanism is $q_c$-responsive, simply induced from its priority order: for all subset of students $S_1\subseteq S$, $$\Ch_c^{\noAA}(S_1) \coloneqq \max(S_1, >_c, q_c).$$ We denote by $\mu^{\noAA} \coloneqq \ttup{SDA}(I, \Ch^{\noAA})$ the matching under the baseline mechanism. Although this matching can be obtained from the original and simpler deferred acceptance algorithm proposed by~\citet{gale1962college}, we present the mechanism from a choice function point of view so that it is consistent with later sections.

\subsection{Discovery program} \label{sec:DISC}

This mechanism is adapted from the policy used by NYC DOE for increasing the number of disadvantaged students at the city's eight specialized schools, which are considered to be the best public schools. The discovery program mechanism distributes reserved seats to disadvantaged student at the end of seat-assignment procedure. One of the reasons for allocating reserved seats to lower ranked disadvantages students is that disadvantaged students who are admitted via reserved seats are required to participate in a 3-weeks \emph{summer enrichment program} as a preparation for the specialized high schools. 

However, for the sake of comparison (with other mechanisms), we assume that students' preference for schools are not affected by whether they are required to participate in the summer program -- that is, students are indifferent between general and reserved seats at each school. We assume this \emph{school-over-seat} hypothesis for the rest of the paper, and we discuss its validity in the Appendix, Section \ref{sec:res-preference}. See also Section \ref{sec:policy-recommendation} for a further discussion.

When there is a shortage of disadvantaged students, reserved seats could go unassigned under the discovery program mechanism. Although this is usually not of concern in real-world applications, since there are usually more students than available seats, we nevertheless present the discovery program mechanism in a more general case where vacant reserved seats are de-reserved \citep{aygun2020dynamic}.

The algorithm for the discovery program mechanism has three stages. Schools' choice functions at all stages are the simple $q$-responsive choice function $\Ch^{\noAA}$. The mechanism starts by running the deferred acceptance algorithm on instance $(G, >, \q^G)$ to obtain matching $\mu^{\ttup{DISC}}_1$ for the general seats; it then runs the deferred acceptance algorithm on the instance restricted to the disadvantaged students that are not yet assigned $(G[C\cup \{s\in S^m: \mu^{\ttup{DISC}}_1(s)= \emptyset\}], >, \q^R)$ to obtain matching $\mu^{\ttup{DISC}}_2$ for reserved seats; and it lastly runs the deferred acceptance algorithm on the instance restricted to the advantaged students that are not yet assigned $(G[C\cup \{s\in S^M: \mu^{\ttup{DISC}}_1(s) = \emptyset\}, >, \q^E)$ with $q^E_c = q^R_c - |\mu^{\ttup{DISC}}_2(c)| \;\forall c\in C$ to obtain matching $\mu^{\ttup{DISC}}_3$ for vacant reserved seats. The final matching combines the matchings obtained at these three stages: $\mu^{\ttup{DISC}} \coloneqq \mu^{\ttup{DISC}}_1 \dot\cup \mu^{\ttup{DISC}}_2 \dot\cup \mu^{\ttup{DISC}}_3$.

Although the mechanism intends to help disadvantaged students, it could actually hurt them. As we show in Example \ref{ex:disc-all-worse}, under the discovery program mechanism, it is possible that all disadvantaged students are worse off. Moreover, the discovery program mechanism could create blocking pairs for disadvantaged students, incentivize disadvantaged students to misrepresent their preference lists or to under-perform, and might hurt disadvantaged students even when the reservation quotas are a smart reserve (see Example \ref{ex:disc-fair-sp}).
\begin{example} \label{ex:disc-all-worse}
    Consider the instance with students $S^M= \{s^M_1, s^M_2\}$, $S^m= \{s^m_1\}$ and schools $C= \{c_1, c_2\}$. The quotas of schools are $q_{c_1}=2$ and $q_{c_2}=1$, and both schools have priority order $s^M_1 > s^M_2 > s^m_1$. Both advantaged students prefer $c_1$ to $c_2$, whereas the disadvantaged student prefers $c_2$ to $c_1$. It is easy to see that under the baseline mechanism, $$\mu^{\noAA} = \{(s^M_1, c_1), (s^M_2, c_1), (s^m_1, c_2)\}.$$ Now consider the discovery program mechanism with reservation quotas $q_{c_1}^R=1$ and $q_{c_2}^R=0$. Then, $$\mu^{\ttup{DISC}} = \{(s^M_1, c_1), (s^M_2, c_2), (s^m_1, c_1)\}.$$ Under the discovery program mechanism, the disadvantaged student $s^m_1$ is not only assigned to a school less preferred less, but is also now required to participate in the summer program. \EOE
\end{example}

\begin{example} \label{ex:disc-fair-sp}
    Consider the instance with students $S^M= \{s^M_1, s^M_2, s^M_3\}$, $S^m= \{s^m_1, s^m_2, s^m_3\}$ and schools $C= \{c_1, c_2\}$. The quotas of schools are $q_{c_1}=3$ and $q_{c_2}=2$, and both schools have priority order $s^M_1 > s^M_2 > s^m_1 > s^M_3 > s^m_2 > s^m_3$. All students prefer $c_1$ to $c_2$. We have $$\mu^{\noAA}(c_1) = \{s_1^M, s_2^M, s_1^m\}, \quad \mu^{\noAA}(c_2) = \{s_3^M, s_2^m\}.$$ Now assume that the reservation quotas are $q^R_{c_1} = q^R_{c_2} = 1$, which in particular is a smart reserve. Under the discovery program mechanism with these reservation quotas, we have $$\mu^{\ttup{DISC}}(c_1) = \{s_1^M, s_2^M, s_2^m\}, \quad \mu^{\ttup{DISC}}(c_2) = \{s_1^m, s_3^m\}.$$ Disadvantaged student $s_1^m$ is worse off under $\mu^{\ttup{DISC}}$ than under $\mu^{\noAA}$. In addition, $\mu^{\ttup{DISC}}$ admits a blocking pair $(s_1^m, c_1)$ as $s_1^m$ prefers $c_1$ to $c_2$ and $s_1^m$ has a higher priority than $s_2^m$ at $c_1$. 
    
    One can see from this example that the discovery program neither is strategy-proof, nor it respects improvements: If $s_1^m$ were to report the preference list as $c_1> \emptyset$ or if $s_1^m$ were to under-perform and reduce their priority standing by one spot (i.e., switch their priority standing with $s_3^M$), the matching under the discovery program mechanism would have been the same as $\mu^\noAA$. \EOE
\end{example}

\subsection{Minority reserve} \label{sec:MR}

Under minority reserve, the choice function of every school $c\in C$, denoted by $\Ch^{\ttup{MR}}_c$, is defined as follows~\citep{hafalir2013effective}: for every subset of students $S_1\subseteq S$, $$\Ch^{\ttup{MR}}_c(S_1) = \underbrace{\max(S_1 \cap S^m, >_c, q_c^R)}_{\eqqcolon S_1^R; \; \textup{reserved seats}} \; \dot\cup \; \underbrace{\max \left( S_1 \setminus S_1^R, >_c, q_c- |S_1^R|) \right)}_{\textup{remaining seats}}.$$ That is, every school first accepts disadvantaged students from its pool of candidates up to its reservation quota, and then fills up the remaining seats from the remaining candidates. Note that if there is a shortage of disadvantage students (i.e., $|S_1\cap C^m| < q_c^R$), then the remaining reserved seats become open to advantaged students. 

\begin{proposition} \label{prop:mr-choice}
    Choice function $\Ch^{\ttup{MR}}_c$ is substitutable, consistent, and $q_c$-acceptant.
\end{proposition}

Since substitutability and consistency guarantee the existence of stable matchings~\citep{aygun2013matching,hatfield2005matching,roth1984stability}, stable matchings exist under choice functions $\Ch^\ttup{MR}$ and we denote by $\mu^{\ttup{MR}} \coloneqq \ttup{SDA}(I ,\Ch^{\ttup{MR}})$ the matching under minority reserve with reservation quotas $\q^R$. Minority reserve has been shown to satisfy several desirable properties, as we summarized in Table \ref{tab:summary-prop}. 

The following claim follows directly from the fact that $\mu^\ttup{MR}$ is stable under choice functions $\Ch^\ttup{MR}$ and the definition of $\Ch^\ttup{MR}$.

\begin{proposition} \label{prop:mr-no-bp}
    $\mu^{\ttup{MR}}$ does not admit in-group blocking pairs.
\end{proposition}

\subsection{Joint seat allocation} \label{sec:JSA}

The mechanism of joint seat allocation we discuss here is inspired by the mechanism used for admission to Indian Institutes of Technology~\citep{JoSAA,sonmez2022affirmative}. It allocates the general and reserved seats at the same time, while only allowing disadvantaged students to take the reserved seats when they cannot get admitted via the general seats. Under this mechanism, the choice function of every school $c\in C$, denoted by $\Ch^{\ttup{JSA}}_c$, is defined as follows. For every subset of students $S_1\subseteq S$,
\begin{equation*}
    \resizebox{.98\textwidth}{!}{
    $\Ch^{\ttup{JSA}}_c(S_1) = \underbrace{\max(S_1, >_c, q_c^G)}_{\eqqcolon S_1^G; \; \textup{general seats}} \; \dot\cup \; \underbrace{\max \left( S_1\cap S^m \setminus S_1^G, >_c, q_c^R \right)}_{\eqqcolon S_1^R; \; \textup{reserved seats}} \; \dot\cup \; \underbrace{\max(S_1\setminus (S_1^G\cup S_1^R), >_c, q_c-|S_1^G\cup S_1^R|)}_{\textup{remainning seats}}.$
    }
\end{equation*}

A prominent distinction between joint seat allocation and minority reserve is that in the former, ``highly ranked'' disadvantaged students are admitted via general seats and do not take up the quotas for reserved seats. Intuitively, this opens up more opportunities for disadvantaged students and one would expect all disadvantaged students to be weakly better off under joint seat allocation than under minority reserve. This is true for instances where the competition for seats is high, but is not true for general instances. See Section \ref{sec:aa-compare} and Theorem \ref{thm:jsa-dom-mr-condition} for more discussions on the comparison between these two mechanisms. 

\begin{proposition} \label{prop:jsa-choice}
    Choice function $\Ch^{\ttup{JSA}}_c$ is substitutable, consistent, and $q_c$-acceptant.
\end{proposition}

Proposition \ref{prop:jsa-choice} implies that stable matchings exist under joint seat allocation, and we denote the student-optimal stable matching by $\mu^{\ttup{JSA}} \coloneqq \ttup{SDA}(I, \Ch^{\ttup{JSA}})$.

On a similarity notes, both minority reserve and joint seat allocation can be viewed as stable matchings under slot specific priorities \citep{kominers2016matching} with vacant seats de-reserved \citep{aygun2020dynamic}, and thus many desirable properties of minority reserve, including weakly strategy-proofness and respect for improvement, also hold for joint seat allocation. We next show additional properties of \ttup{JSA}. See Table \ref{tab:summary-prop} for a complete reference.

\begin{theorem} \label{thm:jsa-one-better}
    For any reservation quota $\q^R$, there exists a disadvantaged student $s\in S^m$ such that $\mu^{\ttup{JSA}}(s) \ge_s \mu^{\noAA}(s)$.
\end{theorem}

\begin{theorem} \label{thm:jsa-smart-all-better}
    If the reservation quotas are a smart reserve, then $\mu^{\ttup{JSA}}$ dominates $\mu^{\noAA}$ for disadvantaged students. 
\end{theorem}

When the reservation quota is not a smart reserve, it is possible that $\mu^{\noAA}$ Pareto dominates $\mu^{\ttup{JSA}}$ for disadvantaged students, which can be readily seen from the same example for minority reserve presented in~\citet{hafalir2013effective}. See also Example \ref{ex:jsa-mr-worse} in Appendix \ref{sec:app:ex:JSA}.

As Proposition \ref{prop:mr-no-bp}, the following claim follows directly from the fact that $\mu^\ttup{JSA}$ is stable under choice functions $\Ch^\ttup{JSA}$ and the definition of $\Ch^\ttup{JSA}$.

\begin{proposition} \label{prop:jsa-no-bp}
    $\mu^{\ttup{JSA}}$ does not admit in-group blocking pairs.
\end{proposition}

\section{Comparison of Mechanisms} \label{sec:aa-compare}

In this section, we investigate how different mechanisms introduced in the previous section compare with each other. All proofs are deferred to the appendix.

\subsection{Is there a winning mechanism for disadvantaged students?}\label{sec:no-one-wins}

To begin with, we would like to answer the following question regarding any two mechanisms: does one mechanism dominate the other mechanism for disadvantaged students? We consider three domains which impose restrictions on the instance or the reservation quotas. They are: (1) the reservation quotas are a smart reserve, (2) schools share a common priority order over the students (i.e., \emph{universal} priority order), and (3) both smart reserve and universal priority order. We summarized the results in Table \ref{tab:summary-compare}. Note that for a pair of mechanisms, a positive answer for (1) or (2) implies a positive answer for (3) and a negative answer for (3) implies negative answers for both (1) and (2). These allow us to simplify the presentations given in Table \ref{tab:summary-compare}.

From Table \ref{tab:summary-compare}, we can see that no two mechanisms are comparable in the general domain (i.e., all instances included). In addition, even in the restricted domains, most of the mechanisms are not comparable, with the exception that minority reserve and joint seat allocation dominate the baseline mechanism when the reservation quotas are a smart reserve. 

These results are shown as follows. We first observe that the baseline mechanism does not dominate the other mechanisms, through a rather trivial example included in Appendix \ref{sec:app:ex:no-one-wins} (see Example \ref{ex:AA-better}). We then compare the \ttup{DISC} with \ttup{MR} and \ttup{JSA} in Example~\ref{ex:disc-no-compare} in Appendix~\ref{sec:app:ex:no-one-wins} and compare \ttup{MR} and \ttup{JSA} in Example~\ref{ex:mr-jsa-no-compare} below.  

We include Example~\ref{ex:mr-jsa-no-compare} in the main body as it shows a rather counterintuitive fact. Since \ttup{JSA} allows top-performing disadvantaged students to take general seats, hence freeing reserved seats for other disadvantaged students, one would expect $\mu^\ttup{JSA}$ to dominate $\mu^\ttup{MR}$ for disadvantaged students \emph{lexicographically} -- that is, we would expect that, when schools share the same ranking of students, the highest ranked disadvantaged student whose school assignment differs between \ttup{JSA} and \ttup{MR} prefers $\mu^\ttup{JSA}$ to $\mu^\ttup{MR}$.  
However, this is not true because of the role played by advantaged students: when disadvantaged students take up general seats under \ttup{JSA}, an advantaged student could become rejected by the school that accepts them under \ttup{MR}, and this particular rejection then creates a ``chain of rejections'' that eventually hurts some disadvantaged student.

\begin{example} \label{ex:mr-jsa-no-compare} 
    Consider the instance with students $S^M=\{s^M_1, s^M_2, s^M_3\}$, $S^m=\{s^m_1, s^m_2, s^m_3, s^m_4\}$ and schools $C=\{c_1, c_2, c_3, c_4\}$. The quotas and reservation quotas of schools, and the preference lists of students are given below.
    \[\setlength{\arraycolsep}{3pt}
    \begin{array}[t]{c|cccc}
        c & c_1 & c_2 & c_3 & c_4 \\
        \hline
        q_c & 1 & 1 & 1 & 2 \\
        q_c^R & 0 & 1 & 0 & 1  
    \end{array} 
    \hspace{2cm} 
    \begin{array}[t]{ccccccccc}
        s^M_1 & s^M_2 & s^M_3 & s^m_1 & s^m_2 & s^m_3 & s^m_4 \\
        \hline
        c_2 & c_1 & c_4 & c_2 & c_4 & c_3 & c_4 \\
         & c_3 & c_3 & c_1 & & &
    \end{array}\]
    All schools have priority order $s^M_1 > s^m_1 > s^M_2 > s^m_2 > s^M_3 > s^m_3 > s^m_4$. To see that the reservation quotas is a smart reserve, the matching under the baseline mechanism is $$\mu^{\noAA} = \{s^m_1, c_1\}, \{s^M_1, c_2\}, \{s^M_2, c_3\}, \{s^m_2, c_4\}, \{s^M_3, c_4\}.$$ The matchings under minority reserve and joint seat allocation are:
    \begin{align*}
        \mu^{\ttup{MR}} &= \{s^M_2, c_1\}, \{s^m_1, c_2\}, \{s^m_3, c_3\}, \{s^m_2, c_4\}, \{s^M_3, c_4\}; \\
        \mu^{\ttup{JSA}} &= \{s^M_2, c_1\}, \{s^m_1, c_2\}, \{s^M_3, c_3\}, \{s^m_2, c_4\}, \{s^m_4, c_4\}.
    \end{align*}
    Disadvantaged student $s^m_1$ and $s^m_2$ are indifferent between $\mu^{\ttup{MR}}$ and $\mu^{\ttup{JSA}}$, $s^m_3$ strictly prefers $\mu^{\ttup{MR}}$ to $\mu^{\ttup{JSA}}$, but $s^m_4$ strictly prefers $\mu^{\ttup{JSA}}$ to $\mu^{\ttup{MR}}$. \EOE
\end{example}

\subsection{Joint seat allocation vs minority reserve: the high competitiveness hypothesis} \label{sec:jsa-vs-mr}

To further compare minority reserve and joint seat allocation, we consider a special condition on the market, that we term \emph{high competitiveness of the market}: 
\begin{equation*}
    |\mu^{\ttup{MR}}(c) \cap S^m| \le q_c^R \hbox{ for every school } c\in C.
    \tag{high competitiveness}
\end{equation*} 

Note that this is an ex-post condition that is based on the outcome $\mu^{\ttup{MR}}$ of a specific mechanism, namely minority reserve. The condition asks that minority students not occupy general seats in matching $\mu^{\ttup{MR}}$. We show empirically that the NYC SHS market is highly competitive using admission data in Section \ref{sec:data}. Under the high competitiveness hypothesis, joint seat allocation dominates minority reserve for disadvantaged students. We formalize the statement in Theorem \ref{thm:jsa-dom-mr-condition}. 

\begin{theorem} \label{thm:jsa-dom-mr-condition}
  For highly competitive markets, $\mu^{\ttup{JSA}}$ dominates $\mu^{\ttup{MR}}$ for disadvantaged students.
\end{theorem}
 
High competitiveness can be connected to primitives of the market. Intuitively, it is satisfied when disadvantaged students are systematically performing worse than advantaged students and when there is a shortage of seats at all schools. In other words, this two condition is satisfied if after the initial allocation of reserve seats to top ranked disadvantaged students, the remaining disadvantaged students are not able to compete with the advantaged students for general seats\footnote{High competitiveness is also satisfied in the trivial case when there are so many reserved seats, that all disadvantaged students get one, but this is rarely seen in the real world -- and does not happen in our data from NYC SHSs.}. This condition is not uncommon in markets with limited resources. 

Below we state a rigorous statement connecting primitives of the market and high competitiveness.  We call a market \emph{homogeneously random} if it satisfies the following conditions:
\begin{enumerate}
    \item[a)] Students' preference lists are independent random permutations of the set of schools\footnote{This assumption is aligned with previous work~\cite{knuth1990stable,pittel1989average,pittel1992likely}. It can be relaxed to a more general albeit more technical condition: preference lists of students are independent, and for every student, any two schools have the same probability of being ranked the first in their preference list.};
\item[b)] Schools share the same ranking of students;
\item[c)] Schools have the same quotas $q$ and reservation quotas $q^R$. 
\end{enumerate}

\begin{theorem} \label{thm:cs-highly-competitive}
Consider a family of homogeneously random markets with an increasing number of students and schools. Assume that $q-1>q^R > n \log n$, where $n$ is the number of schools. If, for some $\epsilon \in (0,1)$, the $r_M:=(n\log n +(q-q^R)n\log \log n)$-th ranked advantaged student exists and is ranked above the $r_m:=(1-\epsilon)q^Rn$-ranked disadvantaged student, where rankings of students are within their respective groups, then the market is highly competitive with probability $1-o(1)$. 
\end{theorem}

Let us discuss more in detail the hypothesis from Theorem~\ref{thm:cs-highly-competitive}. The condition $q^R > n \log n$ applies when there are few schools compared to the number of seats, while the condition on the relative rankings of students applies when disadvantaged students perform systematically worse than advantaged students. Although our result is asymptotic and relies on homogeneity assumptions on the number of seats and preferences of students, it is nonetheless useful to see that for the NYC SHS market, the $r_M$-th ranked advantaged student ranks well above the $r_m$-th ranked disadvantaged student. 
See Appendix~\ref{sec:hc-theorems-applications} for detailed calculations.

The hypothesis from Theorem~\ref{thm:cs-highly-competitive} can be investigated within other models from the literature. As an example, we consider a \emph{homogeneously random market with $(\mu_M,\mu_m,\sigma_M,\sigma_m)-$normal potentials}. This is a market that satisfies a), b), c) from the definition of homogeneously random markets, and moreover schools rank students in decreasing values of their \emph{potentials}\footnote{The assumption that students' potentials are sampled from a distribution  follows a recent trend in the literature, see e.g.,~\cite{faenza2020impact} in the school choice setting and~\cite{kleinberg2018selection} in the hiring setting.}. We assume in particular that potentials of students are drawn i.i.d.~from a normal distribution with variance $\sigma_M$ and mean $\mu_M$ (resp.~variance $\sigma_m$ and mean $\mu_m$) for advantaged students (resp.~disadvantaged students). Other distributional assumption are of course possible and lead to similar results. The next theorem states that if the means are far enough then the market is highly competitive with high probability. 

\begin{theorem}\label{thm:hc-potential}
Consider a family of homogeneously random markets with $(\mu_M,\mu_m,\sigma_M,\sigma_m)-$normal potentials and an increasing number of students and schools. Assume that, for all markets in the family, $q-1>q^R > n \log n$, where $n$ is the number of schools. Let $\epsilon > 0$ be constant and let $$p_M:=\frac{n\log n+(q-q^R-1)n\log\log n}{|S^M|} \textup{ and } p_m:=\frac{(1+\epsilon)q^Rn}{|S^m|}$$ be strictly between $0$ and $1$ and bounded away from both. If
\begin{equation}\label{eq:condition-hc-potential}
\mu_M - \mu_m > 0.008(\sigma_M+\sigma_m) + \frac{1}{1.702}\left( \sigma_M\ln (\frac{1}{p_M}-1) -\sigma_m(\frac{1}{p_m}-1) \right).
\end{equation}
hold, then with probability $1-o(1)$ the market is highly competitive.
\end{theorem}

As with Theorem \ref{thm:cs-highly-competitive}, although the result is asymptotic, we find it informative to confirm that data of SHSAT scores from NYC DOE (see Figure~\ref{fig:score-distributions}) easily verify~\eqref{eq:condition-hc-potential}. See Appendix~\ref{sec:hc-theorems-applications} again for detailed calculations.

\section{Data on NYC Specialized High Schools} \label{sec:data}

In this section, we analyze and compare the mechanisms on real-world datasets\footnote{The dataset is under a non-disclosure agreement with NYC DOE.}. There is a total of $12$ anonymized datasets, each for one of the $12$ consecutive academic years from 2005-06 to 2016-17. Entries of each dataset include (1) students' IDs, (2) their scores for the Specialized High School Admissions Test, (3) their (possibly, non-complete) preference lists of these eight specialized schools, (4) their middle schools, (5) which school they are admitted to (which could be empty), and other information that are not relevant for our analysis. See Table \ref{tab:school-name} for a list of specialized high schools. 

\begin{small}
\begin{table}[ht]
    \centering 
    \begin{tabular}{|c|l|}
        \hline
        B & Bronx High School of Science \\
        \hline
        T & Brooklyn Technical High School \\
        \hline
        R & Staten Island Technical High School \\
        \hline
        L & Brooklyn Latin \\
        \hline
        Q & Queens High School for the Sciences at York \\
        \hline
        M & High School of Mathematics, Science and Engineering at City College \\
        \hline
        S & Stuyvesant High School \\
        \hline
        A & High School of American Studies at Lehman College \\
        \hline
    \end{tabular}
    \vspace{-.5em}
    \caption{School code and school name of NYC specialized high schools. \vspace{-2.5em}}
    \label{tab:school-name}
\end{table}
\end{small}


Immediately from the dataset, we can extract the number of students applying for these specialized high schools and the capacities of each schools (i.e., the number of students admitted). On average, about $27,000$ students take the SHSAT exam every year, and among them, about $8,000$ (which is about $30\%$) are disadvantaged students (defined below). In terms of admission, about $5,100$ students receive an offer, out of whom about $820$ (which is about $16\%$) are disadvantaged students.

To label each student as advantaged or disadvantaged, we follow the definition currently used by NYC DOE for the discovery program:

\begin{quote}\it
    To be eligible for the Discovery program, a Specialized High Schools applicant must be one or more of the following: 
    \begin{enumerate}
        \item a student from a low-income household, a student in temporary housing, or an English Language Learner who moved to NYC within the past four years; and
        \item Have scored within a certain range below the cutoff score on the SHSAT; and
        \item Attend a high-poverty school. A school is defined as high-poverty if it has an Economic Need Index (ENI) of at least 60\%.
    \end{enumerate}
\end{quote}

The second condition is related to eligibility, and not specifically to whether a student is disadvantaged, so we do not incorporate that when labeling the students. For the first set of conditions, we use an accompanying dataset which contains students' demographic information. However, since the information given in the dataset are not exactly the same as those specified in the definition, we slightly modify the first condition: ``be one or more of the following: (1) eligible for free or reduced price lunch or has been identified by the Human Resources Administration (HRA) as receiving certain types of public assistance; or (2) an English Language Learner''. For the last condition, we obtain the ENIs of NYC middle schools from a school quality report of academic year 2017-2018, which can be downloaded from the NYC Open Data website\footnote{\url{https://data.cityofnewyork.us/Education/2017-2018-School-Quality-Reports-Elem-Middle-K-8/g6v2-wcvk}}.

To obtain schools' universal priority order $>_C$ over the students, we assign to every student a unique \emph{lottery} number, denoted as $\ell_s$, for tie-breaking. For any two students $s_1, s_2\in S$, $s_1$ has a higher priority than $s_2$ (i.e., $s_1 >_C s_2$) only when $s_1$ has a higher score than $s_2$ or when they have the same score but $\ell_{s_1}< \ell_{s_2}$. This idea of using lottery numbers for tie breaking has been used in practice (see, e.g.,~\citet{abdulkadirouglu2009strategy}).

Combining all components, the final dataset for analysis contains the following information for each student: unique identification number, test score, preference list, indicator for whether they are disadvantaged students, and the lottery number. 

First in Section \ref{sec:res-dis}, we analyze the outcome of the discovery program mechanism under the current guideline, and we provide some additional observations besides the theoretical results in Section \ref{sec:DISC}. We then compare, in Section \ref{sec:res-compare}, the outcomes from all three mechanisms. For most of the experiments, we only include results of the latest academic year, since they are qualitatively similar for all academic years. Full results of all academic years can be found in Appendix \ref{sec:app:figure-all-years}. We also investigate and discuss the school-over-seat hypothesis by analyzing the patterns of students' preference lists, which can be found in Appendix \ref{sec:res-preference}.

\subsection{Results: the discovery program} \label{sec:res-dis}

We start by analyzing the performance of the discovery program mechanism, where the reservation quota of every school $c$ is set to be $q_c^R \coloneqq \lceil q_c\times 20\% \rceil$, since $20\%$ is the number recommended in a proposal by~\citet{SH-proposal}. 

We show empirically that the discovery program admits in-group blocking pairs and does not respect improvements. As we discussed earlier in Table \ref{tab:summary-prop}, the discovery program is the only mechanism that admits in-group blocking pairs. We show that on average there are about $950$ blocking pairs for disadvantaged students every academic year involving about $650$ disadvantaged students (see Figure \ref{fig:block-pair-disc}). Moreover, the discovery program is the only mechanism under which disadvantaged students can be worse-off when compare to \noAA, when considering the changes in rank to matched schools. In particular, this hurts the top-performing disadvantaged students much more, and helps the low-performing disadvantaged students (see Figure \ref{fig:disc-who-worse}). 

Although, in theory, students could truncate their preference lists to attend better schools, it is unclear if this type of behavior appears systematically in the dataset. Across the 12 academic years of data, the lengths of disadvantaged students' preference lists remains quite constant, with the average being about five schools. Similar lengths of preference lists are observed when restricting to top-performing disadvantaged students only. In addition, we fitted a linear regression model to identify the relationship between the lengths of students' preference lists and their priority standings, but the results were inconclusive. We include the detailed analysis results in Appendix \ref{app:sec:no-strategy}. This is not particularly surprising as remarked by \citet{kesten2010school}: ``failure to satisfy dominant-strategy incentive compatibility does not necessarily imply easy manipulability in practice."

\subsection{Results: comparison of three mechanisms} \label{sec:res-compare}

For experiments in this section, we choose the reservation quotas so that they are consistent with the proportion of disadvantaged students in the market: $q_c^R = \lceil q_c\times \frac{|S^m|}{|S^M|+|s^m|} \rceil$, $\forall c\in C$. We choose these reservation quotas simply because they are a reasonable choice and are a smart reserve, and we would like to point out that one could slightly increase or decrease these numbers without affecting the findings in this section qualitatively.

\begin{figure}[tb]
    \centering
    \begin{subfigure}[t]{.48\textwidth}
    \centering
    \includegraphics[width=.8\textwidth]{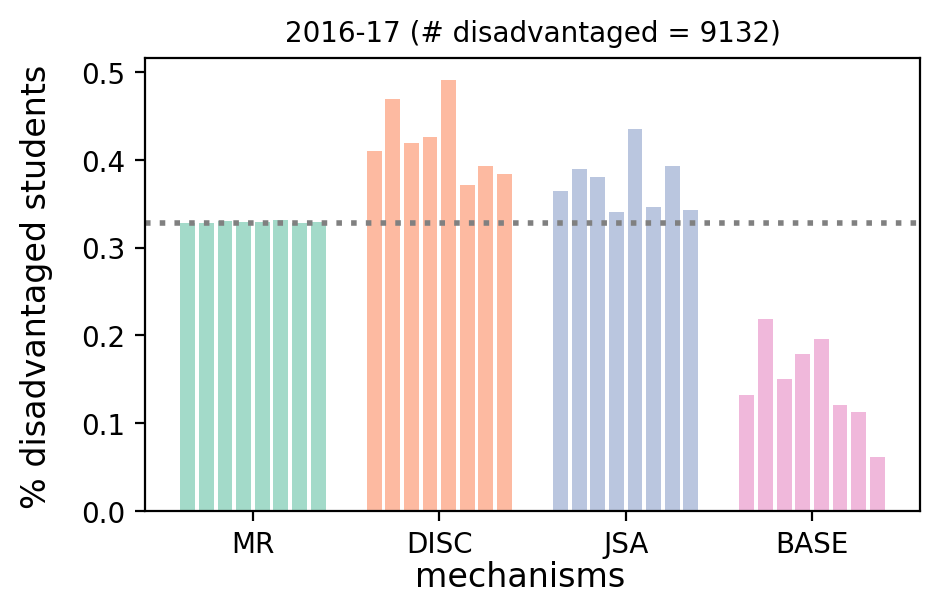}
    \caption{\raggedright \footnotesize Proportions of disadvantaged students admitted, with bars from left to right corresponding to schools: B, T, R, L, Q, M, S, A. The dotted line represents the proportion of disadvantaged students among all applicants.} \label{fig:perc_in_school}
    \end{subfigure} %
    \hspace{.02\textwidth}
    \begin{subfigure}[t]{.48\textwidth}
    \centering
    \includegraphics[width=.8\textwidth]{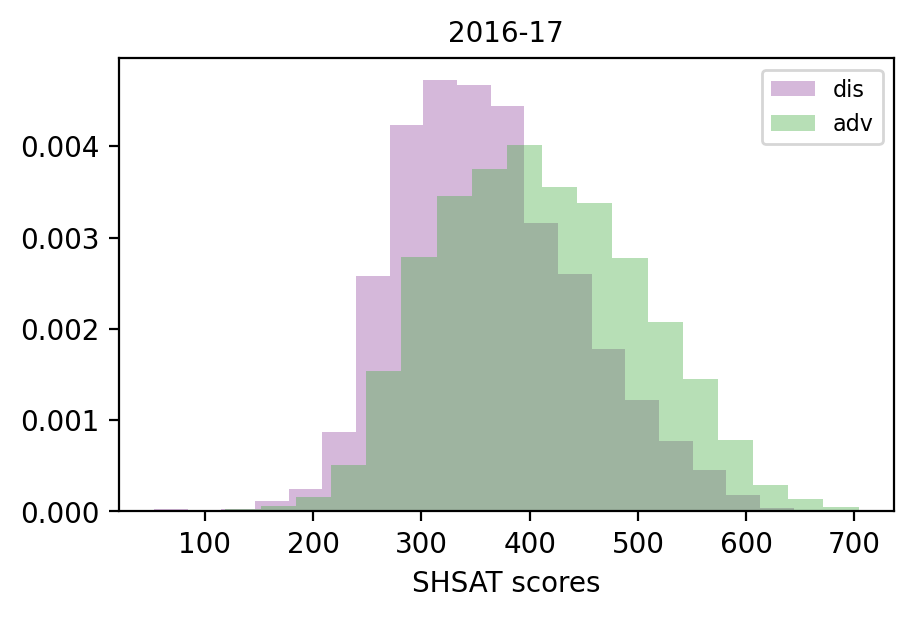}
    \caption{\raggedright \footnotesize The distribution of the SHSAT scores of advantaged students (labeled ``adv'') and disadvantaged students (labeled ``dis'').} \label{fig:score-distributions}
    \end{subfigure}
    \vspace{-1em}
    \caption{Mechanisms with reserved seats increase the number of disadvantaged students admitted.\vspace{-1.5em}}
\end{figure}

\smallskip \noindent\textbf{Proportion of disadvantaged students admitted.} In Figure \ref{fig:perc_in_school}, we show that all mechanisms with reserved seats can increase the proportion of disadvantaged students admitted to these schools. More specifically, under joint seat allocation and the discovery program mechanism, the numbers of disadvantaged students admitted exceeds the reservation quotas. This is because disadvantaged students with high scores can take up general seats under these two mechanisms. On the other hand, for minority reserve, the numbers of disadvantaged students admitted match exactly the reservation quota\footnote{A similar  observation was made by \citet{dur2018reserve} for the Boston school district, where the percentage of walk-zone students hover just around $50\%$ when $50\%$ seats are reserved for them.}. This is because after disadvantaged students take up the reserved seats, the remaining disadvantaged students cannot compete against advantaged students for the general seats and are thus not admitted. The phenomenon is exactly the high competitiveness condition we discussed in Section~\ref{sec:jsa-vs-mr} and is particularly true for our dataset since the number of students are much higher than the number of available seats, and disadvantaged students are in general performing worse than advantaged students, as one can see in Figure \ref{fig:score-distributions}. 

The figure seems to suggest that, for a fixed quota, the discovery program mechanism is better for disadvantaged students, as the number of disadvantaged students admitted to any school is the largest. However, this is not true when we examine the matching more closely down to individual students. Moreover, since one can increase the number of disadvantaged students admitted by simply increasing the reservation quotas, policymakers should not solely focus on the absolute number of disadvantaged students admitted when comparing mechanisms.

\smallskip \noindent\textbf{Effects to individual students.}
As opposed to Figure \ref{fig:perc_in_school} which shows the effects of mechanisms with reserved seats on disadvantaged students as a whole group, we show in Figure \ref{fig:perc_change_rank} these effects on individual levels. In particular, we examine the change in rank of the schools assigned to students under these mechanisms as compared to under the baseline mechanism. For instance, if a student is matched to their third choice (i.e., rank of assigned school is $3$) under the baseline mechanism, but is matched to their first choice (i.e., rank of assigned school is $1$) under minority reserve, then their change in rank of assigned school is $-2$ under minority reserve. 

The main takeaway of Figure \ref{fig:perc_change_rank} is that when the reservation quotas are a smart reserve, the discovery program mechanism is the only one under which disadvantaged students can be worse off, as it is the only mechanism with markers on the positive axis. This is consistent with our discussion in Section \ref{sec:affirmative-action} (see Table \ref{tab:summary-prop}). We further investigate who are the disadvantaged students that are worse off under the discovery program, and we show the results in Figure \ref{fig:disc-who-worse}. Interestingly, the disadvantaged students who are performing relatively well are the ones who are being admitted to schools they prefer less (dots on the upper left side of Figure \ref{fig:disc-who-worse}). These are essentially the disadvantaged students who are assigned to general seats during the first stage of the discovery program mechanism. Because there are fewer seats during the first stage of the discovery program mechanism (as compared to the baseline mechanism), the competition is fiercer and thus, these disadvantaged students got assigned to worse schools. Not only does this phenomenon imply that the discovery program mechanism is unfair to these well-performing disadvantaged students, but it also hints at a situation where students have the incentive to under-perform in the admission exams. This certainly is in sharp contrast to the purpose of education and should not be a consequence of any applicable mechanism. 

\begin{figure}[tb]
    \centering
    \begin{subfigure}[t]{.48\textwidth}
    \centering
    \includegraphics[width=.8\textwidth]{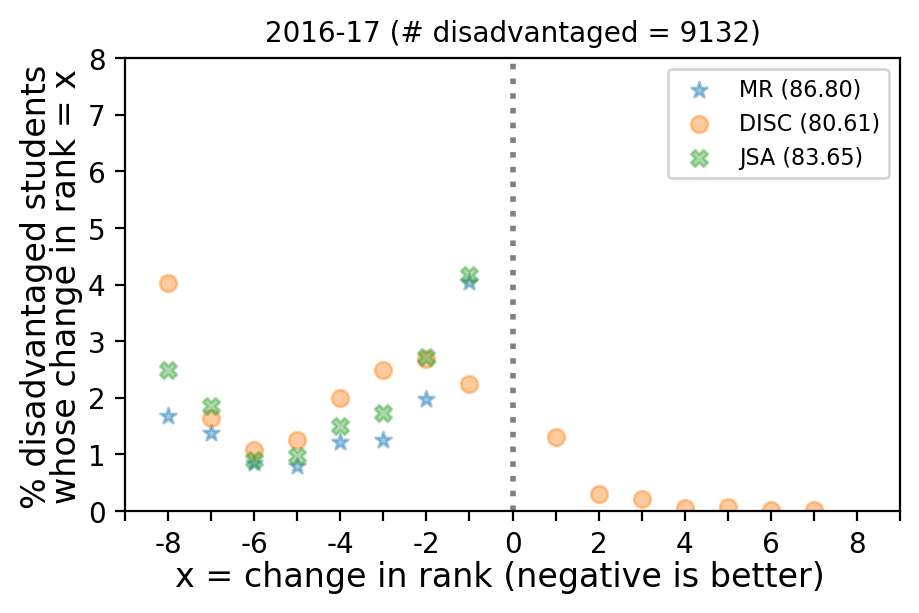}
    \caption{\raggedright \footnotesize Change from \noAA ~to a mechanism with reserved seats, for disadvantaged students} \label{fig:perc_change_rank}
    \end{subfigure}
    \hspace{.02\textwidth}
    \begin{subfigure}[t]{.48\textwidth}
    \centering
    \includegraphics[width=.8\textwidth]{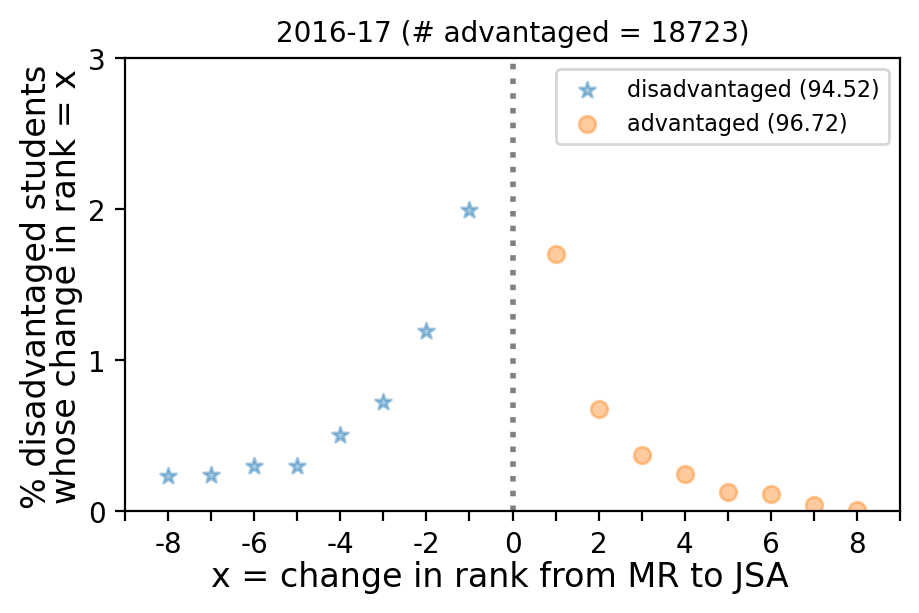}
    \caption{\raggedright \footnotesize Change from \ttup{MR} to \ttup{JSA}, for both advantaged and disadvantaged students.}
    \label{fig:mr-jsa-adv-dis}
    \end{subfigure}
    \vspace{-.5em}
    \caption{\raggedright Percentage of (dis)advantaged students (w.r.t. the total number of (dis)advantaged students) whose change in rank of assigned schools is a certain value. The number in each legend label is for when $x=0$. \vspace{-1.5em}}
\end{figure}

\smallskip \noindent\textbf{Joint seat allocation dominates minority reserve.} In Figure \ref{fig:perc_change_rank}, we see that for each negative change in rank of assigned schools, the markers of joint seat allocation are in general higher than those of minority reserve. It seems to suggest that matching $\mu^{\ttup{JSA}}$ dominates matching $\mu^{\ttup{MR}}$ for disadvantaged students. To understand if this is true, we directly compare these two matchings and confirm the hypothesis (see Figure \ref{fig:mr-jsa-adv-dis}). In fact, we observe the same dominance relation for all academic years. This prompts us to investigate the reason behind it, especially given that this dominance relation is not true in general as we discussed in Section \ref{sec:aa-compare}. This dominance is a consequence of the data satisfying the high competitiveness hypothesis defined in Section~\ref{sec:jsa-vs-mr} (see Figure \ref{fig:perc_in_school}): the number of disadvantaged students admitted under minority reserve should not exceed the reservation quotas.

\section{Conclusion and Discussion} \label{sec:discussion}

In this paper, we study three mechanisms with reserved seats, and compare their outcomes for disadvantaged students under the school-over-seat hypothesis. We show that although the discovery program is instrumental in providing opportunities for disadvantaged students, the current implementation suffers from some drawbacks both theoretically and empirically. Although both joint seat allocation and minority reserve could alleviate these drawbacks, the former is better for disadvantaged students for the NYC specialized high school market. As our main theoretical contribution, we identify a fairly broad condition of markets, that we call \emph{high competitiveness}, under which \ttup{JSA} dominates \ttup{MR} for all disadvantaged students. In particular, we show that this condition holds in the NYC SHS market using 12 years of data.

One caveat of our results is that they are based on the school-over-seat hypothesis, for which current data do not offer a definitive validation. Our experiments on the polarization of the preference data (see Appendix~\ref{sec:res-preference}) and the fact that the length of the summer program (3 weeks) is minimal when compared to the length of a high-school cycle (4 years) seem to suggest that this hypothesis is reasonable. However, other factors may come into play, such as the social stigma attached to being admitted via reserved seats\footnote{We are not aware of this stigma being present in NYC SHSs, but it is definitely present in other markets employing some form of seat reservation \citep{aygun2020designing}.}. We think it is important for the DOE to further investigate this hypothesis, for instance, through questionnaires to the perspective students.

\bigskip
\noindent \textbf{Acknowledgements.} The authors would like to thank the NYC DOE for providing the data as well as further discussion on the discovery program. The authors would also like to express their gratitude to anonymous reviewers of a previous version of the manuscript for their helpful comments and suggestions as well as for pointing out missing references. Yuri Faenza acknowledges support from the NSF Award 2046146 \emph{Career: An Algorithmic Theory of Matching Markets}. A part of Swati Gupta's research has been funded by the NSF AI Institute grant NSF-2112533. Xuan Zhang thanks the Cheung Kong Graduate School of Business (CKGSB) for their fellowship support.

\bibliographystyle{ACM-Reference-Format}
\bibliography{ref}

\newpage\appendix

\section{Missing Definition for Choice Functions} \label{sec:def-ch-func}

\begin{definition}[substitutability]
	Choice function $\Ch_c$ is substitutable if for any set of students $S_1$, $s\in \Ch_c(S_1)$ implies that for all $S_2\subseteq S_1$, $s\in \Ch_c(S_2\cup \{s\})$.
\end{definition}

\begin{definition}[consistency]
	Choice function $\Ch_c$ is consistent if for any sets of students $S_1$ and $S_2$, $\Ch_c(S_1) \subseteq S_2 \subseteq S_1$ implies $\Ch_c(S_1) = \Ch_c(S_2)$.
\end{definition}

\begin{definition}[$q_c$-acceptance]
    Choice function $\Ch_c$ is $q_c$-acceptant if for any set of students $S_1$, $|\Ch_c(S_1)|=\min(q_c, |S_1|)$.
\end{definition}

\begin{definition}[$q_c$-responsive]
    Choice function $\Ch_c$ is $q_c$-responsive if there exists a priority order  $>$ over the students such that for any set of students $S_1$, $\Ch_c(S_1) = \max(S_1, >, q_c)$. In such case, we say $\Ch_c$ is \emph{induced} by priority order $>$ (and quota $q_c$).
\end{definition}

\section{Missing Examples} \label{sec:app:missing-example}

\subsection{From Section~\ref{sec:JSA}} \label{sec:app:ex:JSA}

\begin{example} \label{ex:jsa-mr-worse}
    Consider the instance with students $S^M= \{s^M_1\}$, $S^m= \{s^m_1, s^m_2\}$ and schools $C= \{c_1, c_2, c_3\}$, each with a quota of $1$. All schools have priority order $s^M_1 > s^m_1 > s^m_2$. Students' preference lists are given below:
    \[\setlength{\arraycolsep}{3pt}
    \begin{array}{ccccc}
        s^M_1 & \qquad & s^m_1 & \qquad & s^m_2  \\
        \hline
        c_1 && c_3 && c_1 \\
        c_3 && c_1 && c_2
    \end{array}\]
    Without seat reservation, the resulting matching is $$\mu^{\noAA} = \{(s^M_1, c_1), (s^m_2, c_2), (s^m_1, c_3)\}.$$ Consider the reservation quotas $q_{c_1}^R= 1$ and $q_{c_2}^R= q_{c_3}^R=0$. Then, $$\mu^{\ttup{MR}} = \mu^{\ttup{JSA}} = \{(s^m_1, c_1), (s^m_2, c_2), (s^M_1, c_3)\}.$$ Disadvantaged student $s^m_2$ is indifferent between the two matchings, but disadvantaged student $s^m_1$ strictly prefers $\mu^{\noAA}$ to $\mu^{\ttup{JSA}}$. That is, $\mu^{\noAA}$ Pareto dominates $\mu^{\ttup{JSA}}$ for disadvantaged students. \EOE
\end{example}

\subsection{From Section \ref{sec:no-one-wins}} \label{sec:app:ex:no-one-wins} 

\begin{example} \label{ex:AA-better}
    Consider the instance with students $S^M= \{s_1^M\}$, $S^m= \{s_1^m, s_2^m\}$ and schools $C= \{c_1, c_2\}$. Both schools have a quota of $1$, and a reservation quota of $1$. All students prefer school $c_1$ to $c_2$. Both schools have priority order $s_1^M > s_1^m > s_2^m$. Then, $$\mu^{\noAA} = \{s_1^M, c_1\}, \{s_1^m, c_2\}, \textup{ and } \mu^{\ttup{MR}} = \mu^{\ttup{DISC}} = \mu^{\ttup{JSA}} = \{s_1^m, c_1\}, \{s_2^m c_2\}.$$ That is, the matching under any of the mechanisms with reserved seats Pareto dominates the matching obtained from the baseline mechanism for disadvantaged students. \EOE
\end{example}

\begin{example} \label{ex:disc-no-compare} 
    Consider the instance with students $S^M=\{s^M_1, s^M_2\}$, $S^m=\{s^m_1, s^m_2\}$ and schools $C=\{c_1, c_2\}$. Both schools have a quota of $2$ and a reservation quota of $1$. All students prefer school $c_1$ to $c_2$, and all schools have priority order $s^M_1 > s^m_1 > s^M_2 > s^m_2$. Then, $$\mu^{\noAA} = \mu^{\ttup{MR}} = \mu^{\ttup{JSA}} = \{s^M_1, c_1\}, \{s^m_1, c_1\}, \{s^M_2, c_2\}, \{s^m_2, c_2\},$$ and $$\mu^{\ttup{DISC}} = \{s^M_1, c_1\}, \{s^m_2, c_1\}, \{s^m_1, c_2\}, \{s^M_2, c_2\}.$$ Note that the reservation quotas is a smart reserve. Disadvantaged student $s^m_2$ strictly prefers $\mu^{\ttup{DISC}}$ to the other matching, while $s^m_1$ strictly prefers the other matching to $\mu^{\ttup{DISC}}$. \EOE
\end{example}

\section{Equivalent Interpretation} \label{sec:aux-instances}

In this subsection, we take a different approach and instead of comparing the outputs. We compare how mechanisms interpret the inputs, and particularly how students' original preferences over schools are translated to their preferences over reserved and general seats at all schools. 

\subsection{Techniques}

The mechanisms with reserved seats introduced in this paper seem to entail different algorithms applied to the same preferences lists of students and schools. However, it turns out that an equivalent, yet mathematically more convenient way is to view their assignment outputs as obtained from the same algorithm applied, however, to different input instances (see Section \ref{sec:aux-instances}). There are two approaches by which we can obtain such a reformulation. 

This first approach is to employ \emph{choice functions}, which are a general and powerful way to model the preference lists of agents in matching markets. In particular, all choice functions needed to model the mechanisms in this paper satisfy the \emph{substitutability}, \emph{consistency}, and \emph{$q_c$-acceptance} properties (see Section~\ref{sec:choice-function}). Under such properties, stable matchings are known to exist and satisfy strong structural and algorithmic properties (see, e.g., \citet{alkan2002class,faenza2021affinely,roth1984stability}). This reformulation\footnote{We note in passing, that, this reformulation allows a central planner to access many stable matchings, using recent results by \citet{faenza2021affinely}, which provide alternatives to the matchings output by the mechanisms considered in this paper.} allows us to analyze the assignments under different mechanisms as the outputs of one or more rounds of Roth's generalization \citep{roth1984stability} of the classical deferred acceptance algorithm by \citet{gale1962college}. As a result, to show properties of the assignment obtained from mechanisms with reserved seats, we can directly use properties of its choice functions, of stable matchings, as well as the properties of the generalized deferred acceptance algorithm.

The second approach is to expand students' original preferences over schools to preferences over reserved and general seats at schools. Under this reformulation, assignments under different mechanisms with reserved seats can be obtained simply by applying the classical deferred acceptance algorithm over the equivalent instances. This allows us to deduce interesting properties of the mechanisms (e.g., strategy-proofness), by leveraging on classical results on stable matchings.

\subsection{Auxiliary instances}

We present alternative representations of the inputs under three mechanisms. That is, for each of the three matchings -- $\mu^{\ttup{MR}}$, $\mu^{\ttup{DISC}}$, and $\mu^{\ttup{JSA}}$ -- we show how to construct an auxiliary instance such that the matching corresponds to the student-optimal stable matching of the auxiliary instance without reserved seats.

The reason for developing these auxiliary instances is three-fold. First, it allows us to prove many of the properties (e.g., weakly group strategy-proofness) of the joint seat allocation mechanism, since we can now apply results developed for the classical stable matching model. Second, it completely removes the cost of implementing a new mechanism for the DOE. That is, the DOE does not need to develop a new algorithm incorporating choice functions, and can use the same algorithm as in their current system. Lastly, these auxiliary instances elucidate a simple difference of the three mechanisms: they differ in how students' preferences over general and reserved seats at all schools are extracted from their original preferences over schools.

We start by describing the common components of these auxiliary instances, which are the set of schools, their quotas, and their priority orders over the students. Every school $c\in C$ is divided into two schools $c'$ and $c''$, where $c'$ represents the part with general seats and has quota $q^{\ttup{aux}}_{c'} \coloneqq q_c-q_c^R$, and $c''$ is the part with reserved seats and has quota $q^{\ttup{aux}}_{c''} \coloneqq q_c^R$. Let $C^{\ttup{aux}} = \{c': c\in C\} \cup \{c'': c\in C\}$ be the new set of schools after the division, and for every $c\in C^{\ttup{aux}}$, let $\omega(c)$ denote its corresponding school in the original instance. Then, graph $G^{\ttup{aux}}$ has vertices and edges: $$V(G^{\ttup{aux}})  = C^{\ttup{aux}} \cup S, \textup{ and } E(G^{\ttup{aux}}) = \{(s,c): s\in S, c\in C^{\ttup{aux}} , (s,\omega(c))\in E\}.$$ The priority order over the students by school $c'$ is the same as that of school $c$ (i.e., $>^{\ttup{aux}}_{c'} = >_c$); and that by school $c''$ is defined as follows: for two students $s_1, s_2\in S$, $$s_1 >_{c''}^{\ttup{aux}} s_2 \; \Leftrightarrow \; \begin{cases} s_1\in S^m \textup{ and } s_2\in S^M; \textup{ or } \\ s_1, s_2\in S^m \textup{ and } s_1>_c s_2; \textup{ or } \\ s_1, s_2\in S^M \textup{ and } s_1>_c s_2. \end{cases}$$ The choice function $\Ch^{\ttup{aux}}_{c}$ of every school $c\in C^{\ttup{aux}}$ is $q_c^{\ttup{aux}}$-responsive and is simply induced from priority order $>^{\ttup{aux}}_{c}$. We state the choice functions here to be consistent with our approach in previous sections. However, they are not necessary to obtain the student-optimal stable matching as the classical deferred acceptance algorithm suffice.

The only component remaining is the preference lists of students, which depends on the mechanism with reserved seats, and we describe those next. 

\smallskip \noindent\textbf{Minority reserve.} The original preference list $c_1 >_s c_2 >_s \cdots >_s c_k$ of student $s$ is modified as: $$c''_1 >^{\ttup{MR-a}}_s c'_1 >^{\ttup{MR-a}}_s c''_2 >^{\ttup{MR-a}}_s c'_2 >^{\ttup{MR-a}}_s \cdots >^{\ttup{MR-a}}_s c''_k >^{\ttup{MR-a}}_s c'_k.$$ Although the relative ranking of the schools remains the same, students prefer reserved seats to general seats. Let $I^{\ttup{MR-a}} \coloneqq (G^{\ttup{aux}}, >^{\ttup{MR-a}}_S, >^{\ttup{aux}}_C, \q^{\ttup{aux}})$ denote the auxiliary instance, and let $\mu^{\ttup{MR-a}} \coloneqq \ttup{SDA}(I^{\ttup{MR-a}}, \Ch^{\ttup{aux}})$ denote the student-optimal stable matching of the auxiliary instance.

\begin{proposition}[\cite{hafalir2013effective}] \label{prop:mr-equiv}
    For every student $s\in S$, $\mu^{\ttup{MR}}(s) = \omega(\mu^{\ttup{MR-a}}(s))$.
\end{proposition} 

\noindent\textbf{Discovery program.} The original preference list $c_1 >_s c_2 >_s \cdots >_s c_k$ of student $s$ becomes: $$c'_1 >^{\ttup{DISC-a}}_s c'_2 >^{\ttup{DISC-a}}_s \cdots >^{\ttup{DISC-a}}_s c'_k >^{\ttup{DISC-a}}_s c''_1 >^{\ttup{DISC-a}}_s \cdots >^{\ttup{DISC-a}}_s c''_k.$$ Students prefer general seats over reserved seats; and within each type of seats, the ranking of the schools is the same as that of the original instance. Similarly, we denote the auxiliary instance by $I^{\ttup{DISC-a}} \coloneqq (G^{\ttup{aux}}, >^{\ttup{DISC-a}}_S, >^{\ttup{aux}}_C, \q^{\ttup{aux}})$, and let $\mu^{\ttup{DISC-a}} \coloneqq \ttup{SDA}(I^{\ttup{DISC-a}}, \Ch^{\ttup{aux}})$ denote the student-optimal stable matching of the auxiliary instance. 

\begin{proposition} \label{prop:disc-equiv}
    For every student $s\in S$, $\mu^{\ttup{DISC}}(s)= \omega(\mu^{\ttup{DISC-a}}(s))$.
\end{proposition}

\begin{proof}[Proof of Proposition \ref{prop:disc-equiv}.]
   To prove the proposition, instead of carrying out the deferred acceptance algorithm as we introduced in Section \ref{sec:model-notation} based on~\cite{roth1984stability} for choice function models, we consider an equivalent execution of the algorithm when choice functions $\Ch$ are responsive. This algorithm was introduced by~\citet{mcvitie1971stable} and it similarly runs in rounds. The algorithm starts with all students unmatched. In every round, one student $s$ who is not (temporarily) matched applies to his or her most preferred school $c$ that has not yet rejected him or her. Let $S_c$ denote the set of students $c$ has temporarily accepted at the end of the previous round. School $c$ temporarily accepts $\Ch_c(S_c\cup \{s\})$ and rejects the rest. Note that during the algorithm, at every round, the student $s$ can be arbitrarily selected. Hence, we now consider a particular execution of the algorithm on the auxiliary instance (i.e., the order in which students are selected). The execution has three stages, and they match exactly to the three stages of the discovery program mechanism. In the first stage, the algorithm can only select students who would apply to schools of type $c'$. Since after this stage, students will only apply to schools of type $c''$, the students who are temporarily matched in the first stage would not be rejected in later stages. That is, the temporary assignment at the end of the first stage becomes permanent, and it is matching $\mu^{\ttup{DISC}}_1$. For the second stage, the algorithm can only select disadvantaged students. Since schools of type $c''$ prefers disadvantaged students to advantaged students, the temporary assignment at the end of the second stage is also permanent and it corresponds to $\mu^{\ttup{DISC}}_2$. In the last stage, the algorithm continues without restriction until it terminates. Since there are only advantaged students applying to schools of type $c''$ at this final stage, the matching finalized at this stage is $\mu^{\ttup{DISC}}_3$. 
\end{proof}

\noindent\textbf{Joint seat allocation.} The original preference list $c_1 >_s c_2 >_s \cdots >_s c_k$ of student $s$ becomes: $$c'_1 >^{\ttup{JSA-a}}_s c''_1 >^{\ttup{JSA-a}}_s c'_2 >^{\ttup{JSA-a}}_s c''_2 >^{\ttup{JSA-a}}_s \cdots >^{\ttup{JSA-a}}_s c'_k >^{\ttup{JSA-a}}_s c''_k.$$ Similar to minority reserve, the relative ranking of the schools remains the same as that of the original instance; but different from minority reserve, students prefer general seats to reserved seats. Again, we let $I^{\ttup{JSA-a}} \coloneqq (G^{\ttup{aux}}, >^{\ttup{JSA-a}}_S, >^{\ttup{aux}}_C, \q^{\ttup{aux}})$ denote the auxiliary instance, and let $\mu^{\ttup{JSA-a}} \coloneqq \ttup{SDA}(I^{\ttup{JSA-a}}, \Ch^{\ttup{aux}})$ denote the student-optimal stable matching of the auxiliary instance.

\begin{proposition} \label{prop:jsa-equiv}
    For every student $s\in S$, $\mu^{\ttup{JSA}}(s) = \omega(\mu^{\ttup{JSA-a}}(s))$.
\end{proposition}

\begin{proof}[Proof of Proposition \ref{prop:jsa-equiv}.]
    We first show that matchings in the original instance $I_1 \coloneqq (G, >, \q)$ and matchings in the auxiliary instance $I_2 \coloneqq (G^{\ttup{aux}}, >_S^{\ttup{JSA-a}}, >_C^{\ttup{aux}}, \q)$ can be transformed from each other. One direction is straightforward. Given a matching $\mu_2$ in instance $I_2$, its corresponding matching $\mu_1$ in instance $I_1$ has $\mu_1(s) = \omega(\mu_2(s))$ for all students $s\in S$. For the other direction, let $\mu_1$ be a matching in instance $I_1$, we can construct its corresponding matching $\mu_2$ in instance $I_2$ as follows. For every school $c$, $\mu_2(c') = \max(\mu_1(c), >_c, q_c^G)$ and $\mu_2(c'') = \mu_1(c) \setminus \mu_2(c')$. Let $\psi$ denote the above mapping from matchings in $I_2$ to matchings in $I_1$, and let $\psi^{-1}$ denote the above mapping for the reverse direction. By construction, a matching $\mu$ of $I_1$ is stable in $I_1$ if and only if $\psi^{-1}(\mu)$ is stable in $I_2$. Therefore, the student-optimal stable matching in $I_1$ can be obtained from the student-optimal stable matching in $I_2$ via mapping $\psi^{-1}$, and the claim follows. 
\end{proof}

\section{Missing Proofs} \label{sec:app:missing-proof}

\subsection{From Section \ref{sec:MR}} \label{sec:app:missing-proof-MR}

\begin{proof}[Proof of Proposition \ref{prop:mr-choice}.]
    The substitutability property was shown in~\cite{hafalir2013effective}, but we include the proof here for completeness. Let $S_1\subseteq S$ be a subset of students, $s\in \Ch_c^{\ttup{MR}}(S_1)$ be a student selected by the choice function, and $S_2$ be a subset of students such that $s\in S_2\subseteq S_1$. We want to show that $s\in \Ch^{\ttup{MR}}_c(S_2)$. Consider the following two cases. The first case is when $s\in S_1^R$. Here, it is immediate that $s\in S_2^R\coloneqq \max(S_2\cap S^m, >_c, q_c^R)$ since $S_2\cap S^m\subseteq S_1\cap S^m$ and thus, $s\in \Ch_c^{\ttup{MR}}(S_2)$. The other case is when $s\in \Ch_c^{\ttup{MR}}(S_1) \setminus S_1^R$. Our argument for the first case implies that $S_1^R\cap S_2\subseteq S_2^R$ and thus, we have $S_2 \setminus S_2^R \subseteq S_2 \setminus S_1^R \subseteq S_1 \setminus S_1^R$. Hence, we also have $s\in \Ch_c^{\ttup{MR}}(S_2)$.
    
    Next, for consistency, let $S_2$ be a subset of students with $\Ch_c^{\ttup{MR}}(S_1)\subseteq S_2 \subseteq S_1$, and we want to show that $\Ch_c^{\ttup{MR}}(S_1) = \Ch_c^{\ttup{MR}}(S_2)$. By the definition of the choice function, it is clear that $S_1^R = S_2^R$ since $S_1^R\subseteq S_2$. With the same reasoning, we additionally have $\max(S_1 \setminus S_1^R, >_c, q_c-|S_1^R|)) = \max(S_2 \setminus S_1^R, >_c, q_c- |S_1^R|)) = \max(S_2 \setminus S_2^R, >_c, q_c- |S_2^R|)).$ Therefore, the claim follows. 
    
    Lastly, for $q_c$-acceptance, we first have that $|\Ch_c^{\ttup{MR}}(S_1)| \le |S_1^R| + q_c - |S_1^R| = q_c$, where the inequality follows directly from the definition. It remains to show that when $|S_1| < q_c$, we have $\Ch_c^{\ttup{MR}}(S_1) = S_1$. This is immediate from the definition of the choice function. 
\end{proof}

\begin{proof}[Proof of Propsoition \ref{prop:mr-no-bp}.]
    Assume by contradiction that $(s,c)$ is an in-group blocking pair of $\mu^{\ttup{MR}}$. Let $s'$ be the student in the same group as $s$ such that $s'\in \mu^{\ttup{MR}}(c)$ and $s>_c s'$. Then, by definition of $\Ch^{\ttup{MR}}_c$, we have $s\in \Ch^{\ttup{MR}}_c(\mu^{\ttup{MR}}(c) \cup \{s\})$, which means $(s,c)$ is a blocking pair of $\mu^{\ttup{MR}}$. However, this contradicts stability of $\mu^{\ttup{MR}}$. 
\end{proof}

\subsection{From Section \ref{sec:JSA}}  \label{sec:app:missing-proof-JSA}

\begin{proof}[Proof of Proposition \ref{prop:jsa-choice}.]
    The proof steps are similar to that of Proposition \ref{prop:mr-choice} for minority reserve. Let $S_1\subseteq S$ be a subset of students. First, for substitutability, let $s\in \Ch_c^{\ttup{JSA}}(S_1)$ and let $S_2$ be a subset of students such that $s\in S_2\subseteq S_1$. We want to show that $s\in \Ch_c^{\ttup{JSA}}(S_2)$ and we consider the following three cases. The first case is when $s\in S_1^G$. In this case, it is immediate that $s\in S_2^G\coloneqq \max(S_2, >_c, q_c^G)$ since $S_2\subseteq S_1$. This first case in particular implies that $S_1^G \cap S_2\subseteq S_2^G$ and thus, $S_2\setminus S_2^G \subseteq S_2\setminus S_1^G \subseteq S_1\setminus S_1^G$. Hence, in the second case where $s\in S_1^R$, we similarly have $s\in S_2^R\coloneqq \max( S_2\cap S^m \setminus S_2^G, >_c, q_c^R)$. Note that this argument for the second case also implies that $S_2 \setminus (S_2^G \cup S_2^R) \subseteq S_1\setminus (S_1^G \cup S_1^R)$. Hence, for the last case where $s\in \max(S_1\setminus (S_1^G\cup S_1^R), >_c, q_c-|S_1^G\cup S_1^R|)$, we also have $s\in \max(S_2\setminus (S_2^G\cup S_2^R), >_c, q_c-|S_2^G\cup S_2^R|)$. Therefore, in all these three cases, we have $s\in \Ch_c^{\ttup{JSA}}(S_2)$ and thus $\Ch^{\ttup{JSA}}_c$ is substitutable.
    
    Next, for consistency, let $S_2$ be a subset of students with $\Ch_c^{\ttup{JSA}}(S_1)\subseteq S_2 \subseteq S_1$, and we want to show that $\Ch_c^{\ttup{JSA}}(S_1) = \Ch_c^{\ttup{JSA}}(S_2)$. By the definition of the choice function, it is clear that $S_1^G = S_2^G$ since $S_1^G\subseteq S_2$. Moreover, we have $S_1^R=S_2^R$ since $S_1^R\subseteq S_2\cap S^m\setminus S_2^G$. With the same reasoning, we additionally have that $\max(S_1\setminus (S_1^G\cup S_1^R), >_c, q_c-|S_1^G\cup S_1^R|) = \max(S_2\setminus (S_2^G\cup S_2^R), >_c, q_c-|S_2^G\cup S_2^R|)$. Therefore, the choice function is consistent. 
    
    Lastly, for $q_c$-acceptant, we first have that $|\Ch_c^{\ttup{JSA}}(S_1)| \le |S_1^G| + |S_1^R| + q_c - |S_1^G| - |S_1^R| = q_c$, where the inequality follows directly from the definition. It remains to show that when $|S_1| < q_c$, we have $\Ch_c^{\ttup{JSA}}(S_1) = S_1$. This is immediate from the definition of the choice function. 
\end{proof}

\begin{proof}[Proof of Theorem \ref{thm:jsa-one-better}.]
    Assume by contradiction that there is reservation quotas $\q^R$ such that $\mu^{\noAA}(s) >_s \mu^{\ttup{JSA}}(s)$ for every disadvantaged student $s\in S^m$. Then, consider an alternative instance where every disadvantaged student $s$ misreports his or her preference list where $\mu^{\noAA}(s)$ is the only acceptable school. Let $\tilde G$ and $\tilde >_S$ be the resulting graph and preference lists of the students. In the following, we consider the alternative instance $\tilde I = (\tilde G, \tilde >_S, >_C, \q)$ and we claim that $\mu^{\noAA}$ is stable in instance $\tilde I$ under choice functions $\Ch^{\ttup{JSA}}$. Assume by contradiction that $\mu^{\noAA}$ admits a blocking pair $(s,c)$. Since all disadvantaged students are matched to their first choice, it must be that $s\in S^M$. Then, $s\in \Ch^{\ttup{JSA}}_c (\mu^{\noAA}(c) \cup \{s\})$ implies that there is a student $s'\in \mu^{\noAA}(c)$ such that $s>_c s'$. However, this means $s\in \Ch^{\noAA} (\mu^{\noAA}(c) \cup \{s\})$, which contradicts stability of $\mu^{\noAA}$ in the original instance $I$ under choice functions $\Ch^{\noAA}$. Hence, $\mu^{\noAA}$ is stable in instance $\tilde I$ with choice functions $\Ch^{\ttup{JSA}}$. Since $\ttup{SDA}(\tilde I, \Ch^{\ttup{JSA}})$ is the student-optimal stable matching, it dominates $\mu^{\noAA}$ and thus, every disadvantaged student is strictly better off under $\ttup{SDA}(\tilde I, \Ch^{\ttup{JSA}})$ as compared to $\mu^{\ttup{JSA}}$. However, this contradicts the fact that the joint seat allocation mechanism is weakly group strategy-proof \citep{kominers2016matching,aygun2020dynamic}. 
\end{proof}

\begin{proof}[Proof of Theorem \ref{thm:jsa-smart-all-better}.]
    Assume by contradiction that there exists disadvantaged students $s$ with $\mu^{\noAA}(s) >_s \mu^{\ttup{JSA}}(s)$. Let $s_1$ be the \emph{first} disadvantaged student that is rejected by $c_1 \coloneqq \mu^{\noAA}(s_1)$ during the deferred acceptance algorithm on instance $I$ with choice functions $\Ch^{\ttup{JSA}}$. Assume this rejection happens at round $k$. Let $S^{\ttup{JSA}}_k$ denote the set of students who apply to school $c_1$ during round $k$. In addition, let $S^{\noAA}$ denote the set of students who have ever applied to $c_1$ throughout the deferred acceptance on instance $I$ with choice functions $\Ch^{\noAA}$. It has been shown in~\cite{roth1984stability} that $\Ch^{\noAA}_{c_1} (S^{\noAA}) = \mu^{\noAA}(c_1)$. Thus, $s_1 \in \max(S^{\noAA} \cap S^m, >_{c_1}, q_{c_1}^R)$ by definition of choice function $\Ch_{c_1}^{\noAA}$ and the assumption that the reservation quotas are a smart reserve (i.e., $q_{c_1}^R \ge |\mu^{\noAA}(c_1)|$). Moreover, by our choice of $s_1$, we have $S^{\ttup{JSA}}_k \cap S^m \subseteq S^{\noAA} \cap S^m$. Therefore, $s_1 \in \max(S^{\ttup{JSA}}_k \cap S^m, >_{c_1}, q_{c_1}^R)$, which then implies $s_1 \in \Ch^{\ttup{JSA}}_{c_1} (S^{\ttup{JSA}}_k)$ by definition of choice function $\Ch^{\ttup{JSA}}_{c_1}$. However, this contradicts our assumption that $s_1$ is rejected by $c_1$ at round $k$, concluding the proof. 
\end{proof}

\begin{proof}[Proof of Proposition \ref{prop:jsa-no-bp}.]
    Assume by contradiction that $(s,c)$ is an in-group blocking pair. Let $s'$ be the student in the same group as $s$ such that $s'\in \mu^{\ttup{JSA}}(c)$ and $s>_c s'$. Then, by definition of $\Ch^{\ttup{JSA}}_c$, we have $s\in \Ch^{\ttup{JSA}}_c(\mu^{\ttup{JSA}}(c) \cup \{s\})$, which means $(s,c)$ is a blocking pair of $\mu^{\ttup{JSA}}$. However, this contradicts stability of $\mu^{\ttup{JSA}}$. 
\end{proof}

\subsection{From Section \ref{sec:jsa-vs-mr}} \label{sec:app:missing-proof-compare}

\begin{proof}[Proof of Theorem \ref{thm:jsa-dom-mr-condition}.]
    Assume by contradiction there exists disadvantaged students $s$ such that $\mu^{\ttup{MR}}(s) >_s \mu^{\ttup{JSA}}(s)$. Consider the execution of the deferred acceptance algorithm with choice functions $\Ch^{\ttup{JSA}}$, and let $s_1$ be the first disadvantaged student who is rejected by $\mu^{\ttup{MR}}(s_1)\coloneqq c_1$. Assume this rejection happens at round $k$ of the deferred acceptance algorithm. Let $S^{\ttup{JSA}}_{k}$ denote the set of students who apply to school $c_1$ during round $k$. In addition, let $S^{\ttup{MR}}$ denote the set of students who have ever applied to school $c_1$ during the execution of the deferred acceptance algorithm with choice functions $\Ch^{\ttup{MR}}$. It has been shown in~\cite{roth1984stability} that $\Ch^{\ttup{MR}}_{c_1}(S^{\ttup{MR}}) = \mu^{\ttup{MR}}(c_1)$, which then implies that $s_1\in \max(S^{\ttup{MR}} \cap S^m, >_{c_1}, q_{c_1}^R)$ by definition of choice function $\Ch^{\ttup{MR}}_{c_1}$ and our assumption that $|\mu^{\ttup{MR}}(c_1)| \le q_{c_1}^R$. Moreover, our choice of student $s_1$ implies that $S^{\ttup{JSA}}_{k} \cap S^m \subseteq S^{\ttup{MR}} \cap S^m$ and thus, we also have $s_1 \in \max(S^{\ttup{JSA}}_k \cap S^m, >_{c_1}, q_{c_1}^R)$. Therefore, $s_1\in \Ch^{\ttup{JSA}}_{c_1} (S^{\ttup{JSA}}_{k})$ by definition of choice function $\Ch^{\ttup{JSA}}_{c_1}$. However, this contradicts our assumption that $s_1$ is rejected by $c_1$ at round $k$, concluding the proof. 
\end{proof}

\begin{proof}[Proof of Theorem~\ref{thm:cs-highly-competitive}.] Recall that, under \texttt{MR}, a student applies to her favorite school's reserved seats, and, if rejected, to the same school's non-reserved seat (see Section 4). We want to estimate the ranking, among disadvantaged students, of the \emph{bottleneck} student -- that is, the first disadvantaged  student that is not admitted through a reserved seat at her most preferred school (hence, the student may either be admitted to her most preferred school via a general seat, or be admitted to another school, or not be admitted to any school).

We reformulate this problem in the classical balls in bins setting: given $n$ bins and a series of balls, each inserted in exactly one bin chosen uniformly at random, which is the first ball $k$ that is inserted in a bin with already $q^R$ balls? Classical bounds (see, e.g.,~\cite{raab1998balls}) imply that, in the $q^R > n \log n$ regimen, $k\geq (1-\epsilon){q^R}{n}$ with probability $1-o(1)$ for any $\epsilon \in (0,1)$ -- in particular, for the $\epsilon$ from the hypothesis of the theorem. Interpreting schools as bins, disadvantaged students as balls, and assigning students to their most preferred schools as inserting balls to bins, we obtain that, with probability $1-o(1)$, the bottleneck student is ranked at least $(1-\epsilon){q^R}{n}$ among disadvantaged students. 

The market is highly competitive if and only if any disadvantaged student ranked at par or worse than the bottleneck student does not get a general seat in any school. For this to happen, the bottleneck student must be ranked worse than an advantaged student that we call \emph{lucky applicant}. This is the worst-ranked advantaged student that would get a non-reserved seat in the market obtained from the original market with the number of seats being $q-q^R$, no reservation quota, and no disadvantaged student (call such a market \emph{restricted}). So we want to compute the ranking, among advantaged students, of the lucky applicant. We can use again the balls and bins analogy from above. Denote by $b(q-q^R,n)$ the random variable denoting the smallest $p$ such that, when ball $p$ is extracted, all bins already have at least $(q-q^R)$ balls inserted. From~\cite{erdHos1961classical}, we know that for any real $x$, we have
$$\lim_{n\rightarrow \infty}\mathbb{P}(b(q-q^R,n)-1 < n\log n + n(q-q^R-1)\log \log n + nx) =e^{-\frac{e^{-x}}{(q-q^R-1)!}}.
$$
Taking $x=\log \log \log n$, we have 
\begin{eqnarray*}\lim_{n\rightarrow \infty}\mathbb{P}(b(q-q^R,n)-1 < n\log n + n(q-q^R-1)\log \log n + n \log \log \log n) & =& \lim_{n\rightarrow \infty} e^{-\frac{e^{-\log \log \log n}}{(q-q^R-1)!}}\\ & \geq &  \lim_{n\rightarrow \infty} e^{-e^{-\log \log \log n}}\\ & = & 1.\end{eqnarray*}
Hence, with probability $1-o(1)$, each school is ranked first at least $(q-q^R-1)$ times when we look at the preference lists of the best $n\log n + (q-q^R)n\log \log n$ advantaged students. Thus, with high probability, all the advantaged students that are admitted to a seat in the restricted market -- in particular, the lucky applicant -- are contained in the  $(n\log n +(q-q^R)n\log \log n)$-best ranked advantaged students. It suffices therefore that the worst of them is ranked above the bottleneck student -- as it is required by the hypothesis -- to conclude that the market is highly competitive. 
\end{proof}

\begin{proof}[Proof of Theorem~\ref{thm:hc-potential}]
Because of the hypothesis, the $k$-th order statistic of a distribution converges in probability to the $k$-th quantile function of the CDF (see, e.g.,~\cite[Chapter 7]{dasgupta2008asymptotic}). We use the approximation for the quantile function of a standard normal distribution from~\cite{bowling2009logistic}:
$$
\phi^{-1}(\alpha)\approx \frac{1}{-1.702} \ln (\frac{1}{\alpha} -1).
$$
which has an absolute error of at most $1.4 \cdot 10^{-2}$, see~\cite{soranzo2014very}. The approximation of $\phi^{-1}(\alpha)$ can be plugged in the formula of the quantile function of a normal with mean $\mu$ and variance $\sigma$:
$$\phi_{\mu,\sigma}^{-1}(\alpha)=\mu + \sigma \phi^{-1}(\alpha)\approx \mu + \sigma \frac{1}{-1.702} \ln (\frac{1}{\alpha} -1),$$ which has therefore an absolute error of at most $0.8 \sigma \cdot 10^{-2}$. In order to apply Theorem~\ref{thm:cs-highly-competitive}, we need that 
$$
\phi_{\mu_M,\sigma}^{-1}(p_M) > \phi_{\mu_m,\sigma}^{-1}(p_m).
$$
Rearranging and plugging in the approximation above, we obtain~\eqref{eq:condition-hc-potential}.\end{proof}

\section{Applications of Theorem~\ref{thm:cs-highly-competitive} and Theorem~\ref{thm:hc-potential}}\label{sec:hc-theorems-applications}

\subsection{Application of Theorem~\ref{thm:cs-highly-competitive} to data from NYC SHS}\label{sec:hc-theorems-applications-one} The average reservation quota and the average number of seats at each school is respectively $q^R=208$ and $q=635$. In addition, for the ranking, we have $n+n(q-q^R)=3424$, and $q^R n=1664$. Note that we omit from the comparison the terms logarithmic and sublogarithmic in $n$ because $n=8$ and thus these terms would only help the hypothesis of Theorem~\ref{thm:cs-highly-competitive} to be satisfied. We see that the $1664$-th ranked disadvantaged student performs at par with the $6848$-th advantaged student, hence well within the comparative rank condition of Theorem \ref{thm:cs-highly-competitive}.

\subsection{Application of Theorem~\ref{thm:hc-potential} to SHSAT scores}\label{sec:hc-theorems-applications-two}

From data on the SHSAT scores of students (under NDA with the department of Education of NYC), we know that $|S^m|=9132$, $|S^M|=18723$, $\mu^m=362.40$, $\mu^M=408.76$, $\sigma_M=92.53$, $\sigma_m=83.13$, $p_M=.18$, $p_m=.18$. Hence, the left-hand side of~\eqref{eq:condition-hc-potential} (difference between the means) is $46.36$, which is much larger than the right-hand side, which is $9.47$. 

\section{Lack of Evidence for Strategic Behaviors} \label{app:sec:no-strategy}

In Table \ref{tab:lengh-pref-list}, we present the average length of preference lists submitted to the DOE by all disadvantaged students and by the top 500 disadvantaged students across all 12 academic years. Numbers here do not elicit clear evidence for strategic behaviors, as there is neither a pattern of top-performing disadvantaged students truncating their preference lists more than other disadvantaged students, nor is a there a trend that over time more top-performing disadvantaged students start to truncate their preference lists.

\begin{table}[ht]
    \centering
    \begin{tabular}{c|cc}
        academic year & all & top 500 \\
        \hline
        2005-06 & 5.04 & 5.10 \\
        2006-07 & 5.35 & 5.63 \\
        2007-08 & 4.93 & 4.77 \\
        2008-09 & 5.10 & 5.00 \\
        2009-10 & 5.17 & 5.01 \\
        2010-11 & 5.28 & 5.05 \\
        2011-12 & 5.25 & 5.16 \\
        2012-13 & 5.22 & 4.81 \\
        2013-14 & 5.49 & 5.11 \\
        2014-15 & 5.56 & 5.58 \\
        2015-16 & 5.63 & 5.34 \\
        2016-17 & 5.34 & 5.60 
    \end{tabular}
    \caption{Average length of preference lists submitted to the DOE by ``all'' disadvantaged students, and by ``top 500'' disadvantaged students only. }
    \label{tab:lengh-pref-list}
\end{table}

In Table \ref{tab:corr-len-rank}, we summarized the Pearson correlation coefficients and the corresponding p-values between the lengths of students' preference lists and their priority standings, across all 12 academic years. There is no clear evidence that top-performing students are more likely to truncate their preference lists. In fact, for 4 out of the 12 academic years, we observe that top-performing students have significantly longer preference lists, which is contrary to the strategy behaviors. 

\begin{table}[ht]
    \centering
    \begin{tabular}{c|cl}
        academic year & correlation coefficient & p-value \\
        \hline
        2005-06 & -0.0433 & 0.0002 $\dagger$ \\
        2006-07 & -0.0306 & 0.0157 $\dagger$ \\
        2007-08 & -0.0087 & 0.4905 \\
        2008-09 & -0.0269 & 0.0283 $\dagger$ \\
        2009-10 & -0.0042 & 0.7037 \\
        2010-11 & 0.0214 & 0.0467 $\ddagger$ \\
        2011-12 & 0.0221 & 0.059 \\
        2012-13 & 0.0086 & 0.4548 \\
        2013-14 & 0.0281 & 0.0047 $\ddagger$ \\
        2014-15 & -0.0071 & 0.4855 \\
        2015-16 & 0.006 & 0.5661 \\
        2016-17 & -0.0539 & 0.0 $\dagger$  
    \end{tabular}
    \caption{The correlation between the lengths of students' preference lists and their priority standing where higher score implies a smaller number in priority standing. For significant results, we noted the p-values with either $\ddagger$ or $\dagger$, where $\ddagger$ is for cases where the top-performing students have significantly shorter preference lists, and $\dagger$ is for cases where the top-performing students have significantly longer preference lists.} \label{tab:corr-len-rank}
\end{table}

\newpage
\section{Additional Figures for all Academic Years} \label{sec:app:figure-all-years}

\begin{figure}[h]
    \centering
    \includegraphics[width=.8\textwidth]{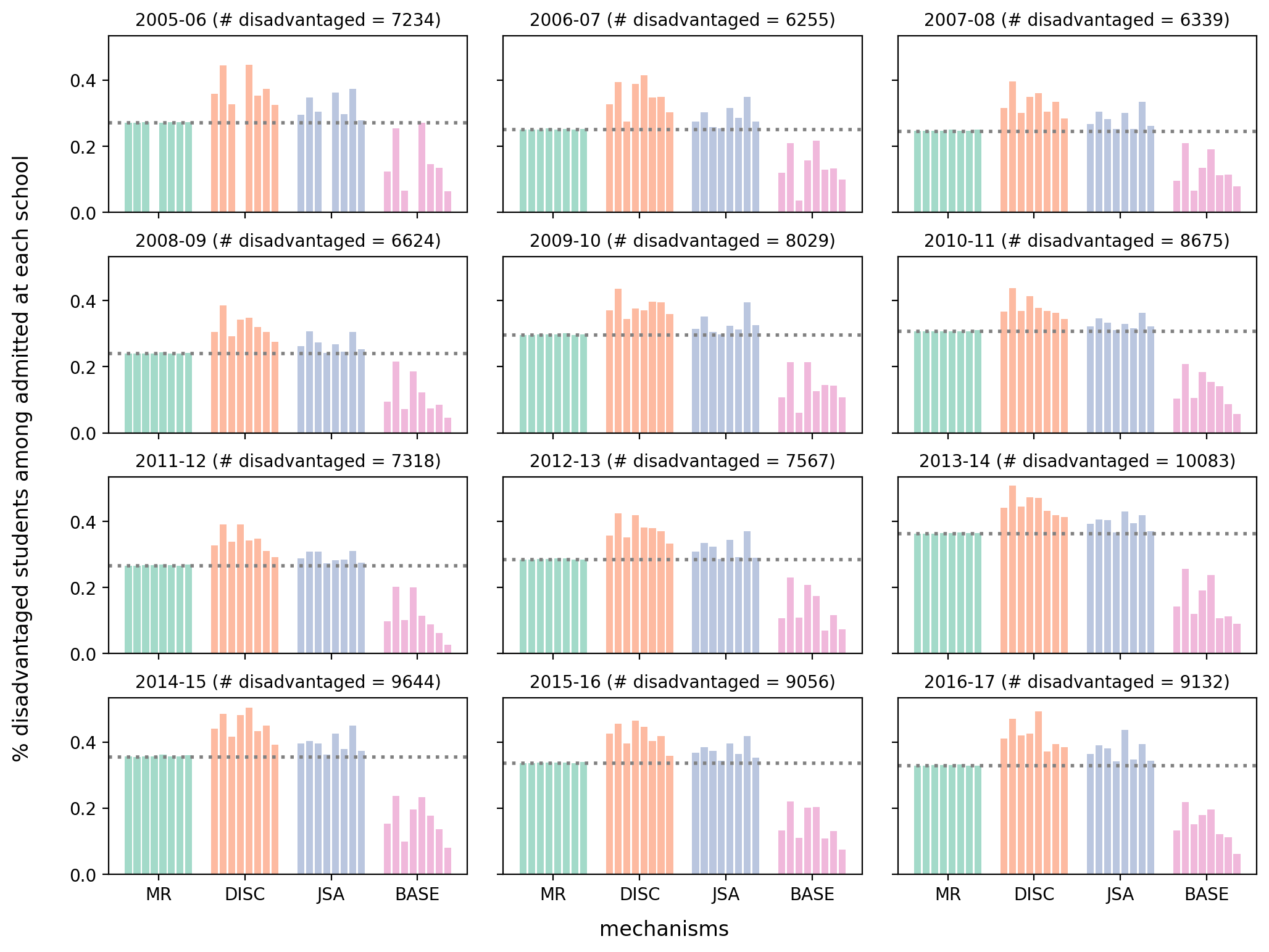}
    \caption{All academic years of Figure~\ref{fig:perc_in_school}. For the 2008-09 academic year, there were only seven SHSs.}
\end{figure}
\text{~}

\newpage
\text{~}
\begin{figure}[t]
    \centering
    \includegraphics[width=.8\textwidth]{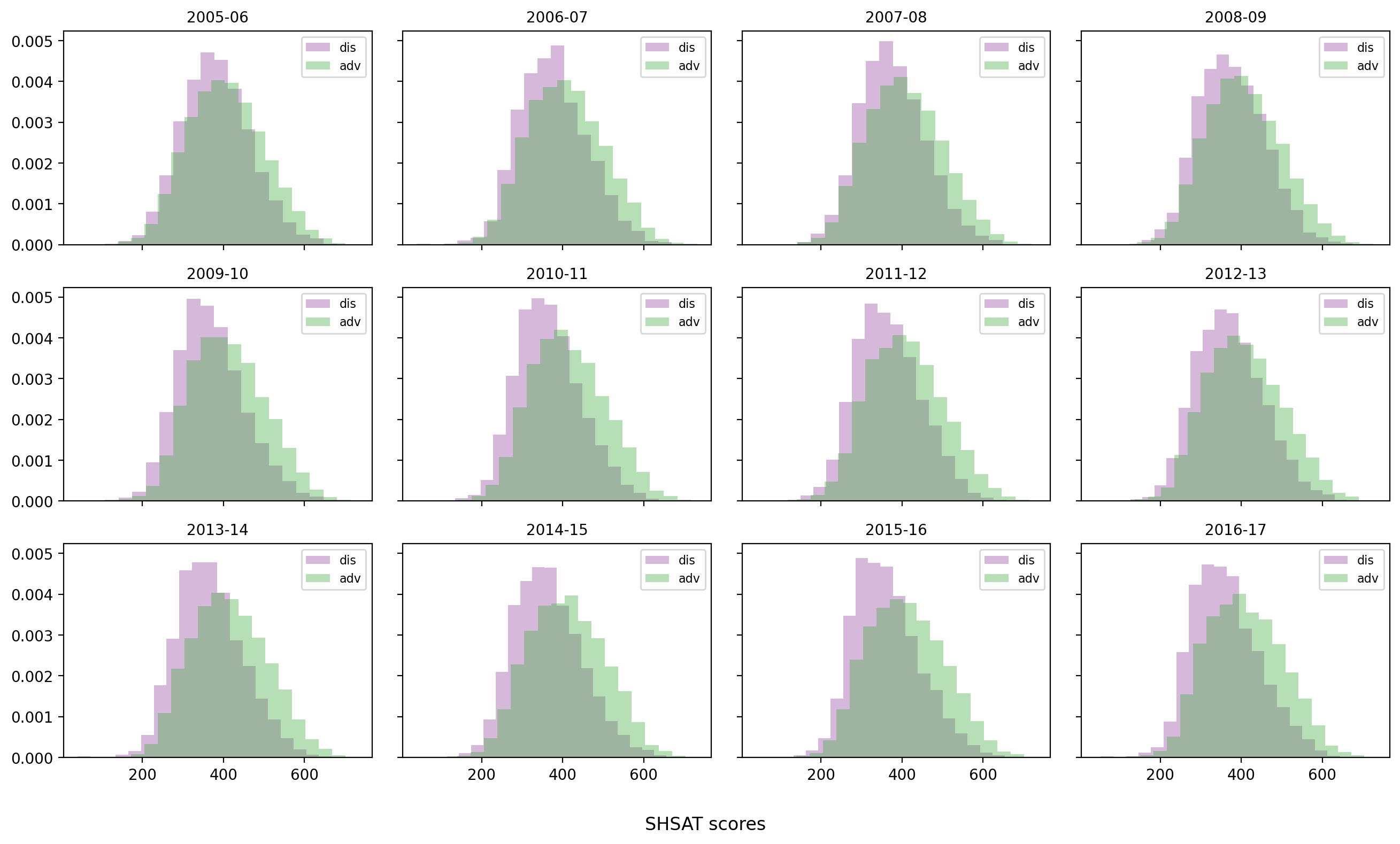}
    \caption{All academic years of Figure~\ref{fig:score-distributions}.}
\end{figure}
\text{~}
\vspace{.1cm}
\text{~}
\begin{figure}[h]
    \centering
    \includegraphics[width=.8\textwidth]{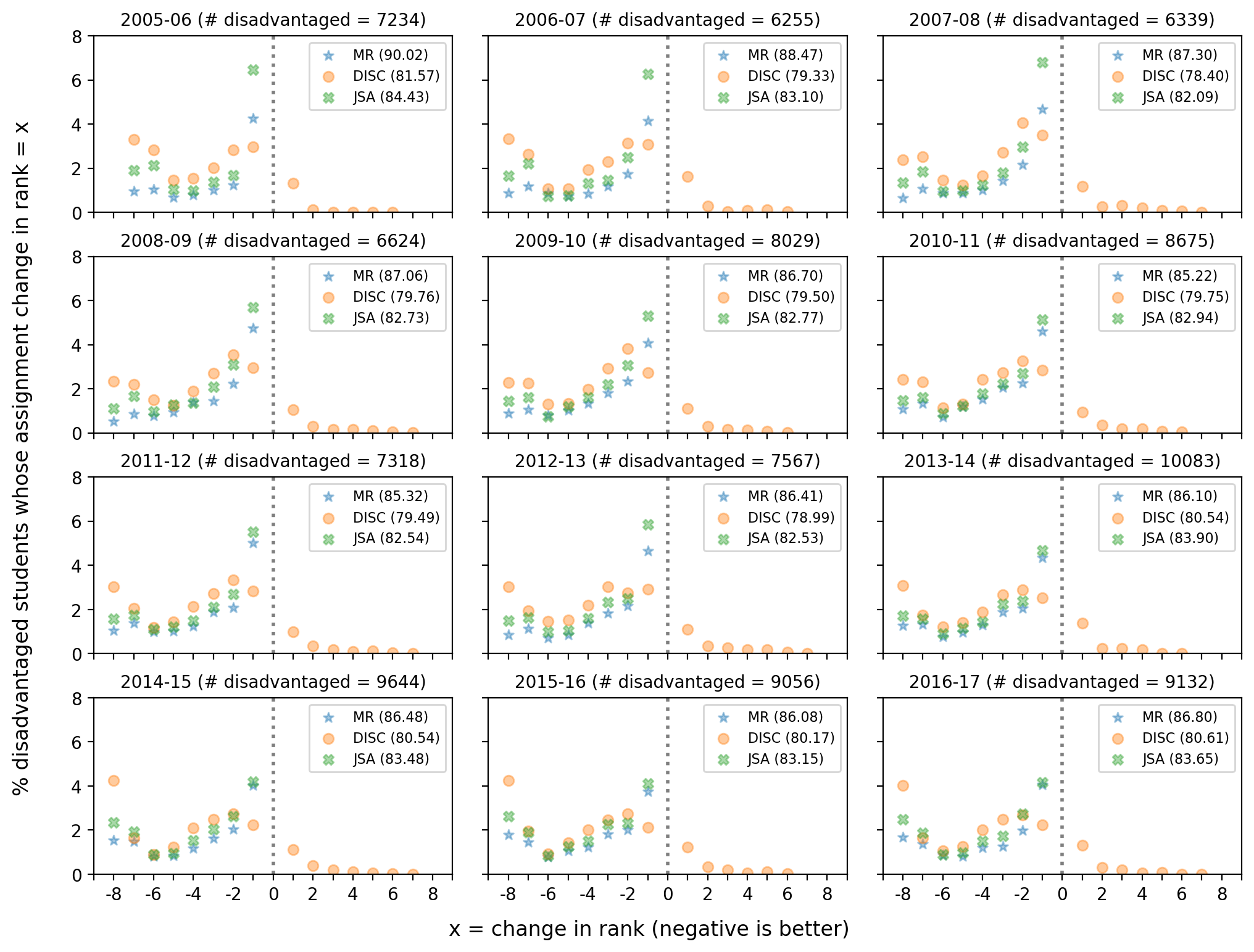}
    \caption{All academic years of Figure~\ref{fig:perc_change_rank}.}
\end{figure}
\text{~}

\newpage
\text{~}
\begin{figure}[t]
    \centering
    \includegraphics[width=.8\textwidth]{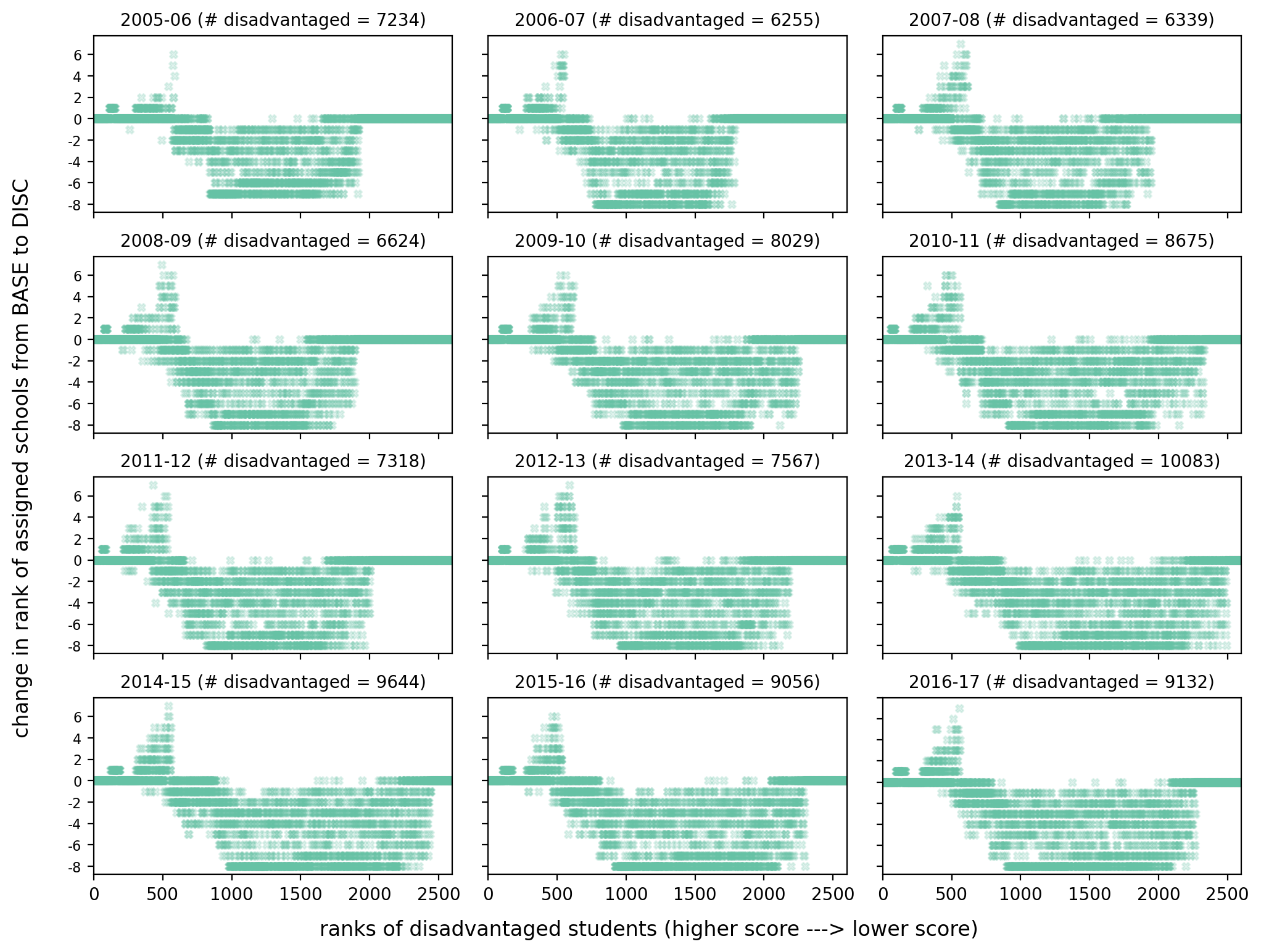}
    \caption{All academic years of Figure~\ref{fig:disc-who-worse}.}
\end{figure}
\text{~}
\vspace{.1cm}
\text{~}
\begin{figure}[h]
    \centering
    \includegraphics[width=.8\textwidth]{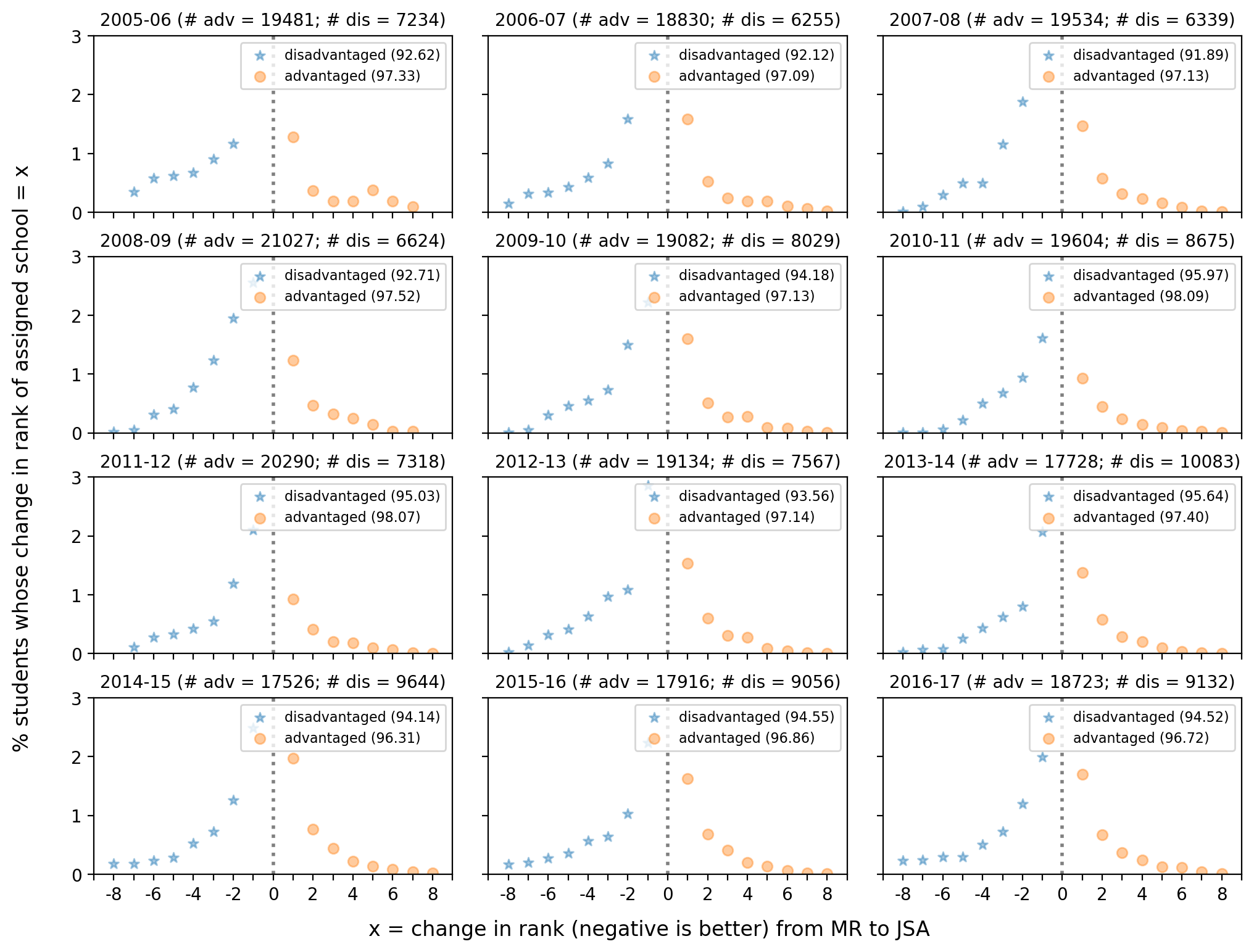}
    \caption{All academic years of Figure~\ref{fig:mr-jsa-adv-dis}.}
\end{figure}
\text{~}

%
%
%
%

\newpage
\section{Discussion on the school-over-seat hypothesis} \label{sec:res-preference}

In this section, we delve into some empirical observations of students' preference lists and we do so for two reasons. The first one is to investigate the school-over-seat hypothesis. Since students are not asked to report their preferences over different types of seats, we can only make some inferences based on the pattern of the preferences submitted by students. For the second reason, recall that in Section \ref{sec:aux-instances}, we show how different mechanisms expand differently students' original preferences over schools to their preferences over reserved and general seats. Hence, our observations aim to shed some light on the validity of these expansions. For the following discussion, we forgo the assumption that participation in the summer enrichment program does not affect students' preference for schools.

\begin{figure}[th]
    \centering
    \includegraphics[width=.7\textwidth]{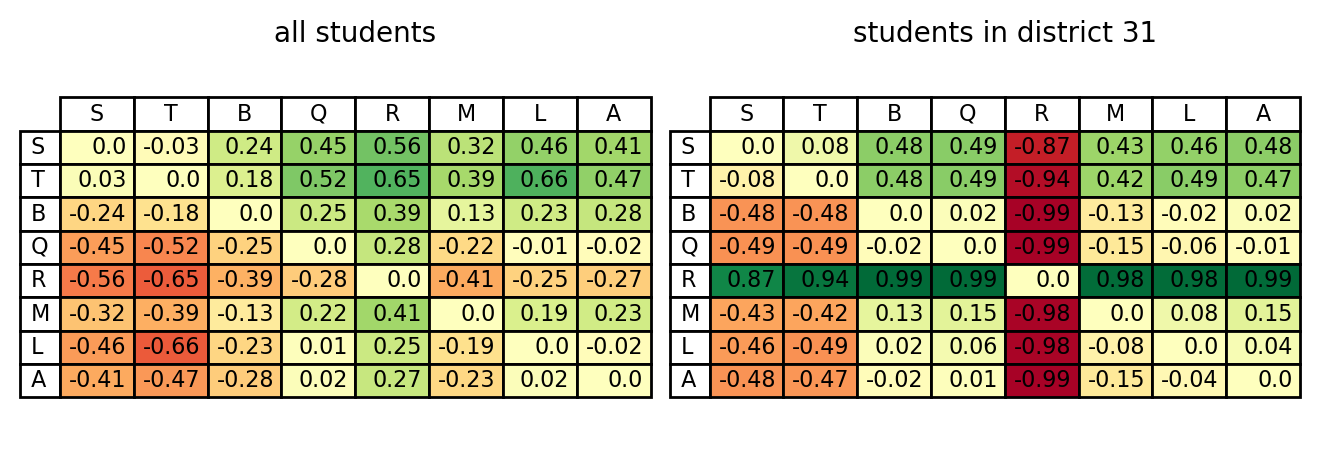}
    \caption{\raggedright Each cell in this table represents the extent to which students prefer the row school to the column school. Specifically, the number is calculated as the percentage of students in each district who prefer the row school to the column school minus the percentage of students who prefer the column school to the row school. The cells are color-formatted with numbers in $[-1,1]$ mapped to a spectrum from red to green. } \label{fig:pwc_sch}
\end{figure}

The second table in Figure \ref{fig:pwc_sch} indicates that geographic proximity could lead to a strong preference for some schools. We observe that students in district 31 strongly prefer Staten Island Tech (S) to any other schools. This is because district 31 is the only school district on Staten Island, and Staten Island Tech is the only specialized high school on Staten Island. Hence, for students residing in Staten Island, since transportation to other boroughs are extremely limited and lengthy, it is reasonable to assume the school-over-seat hypothesis when comparing Staten Island Tech to any other specialized high school. We show in Figures \ref{fig:pwc_sch_all_1} Figure \ref{fig:pwc_sch_all_2} the same type of tables for other school districts, where we observe similar patterns: students in district 10 strongly prefers Bronx Science (B) and students in district 29 strongly prefers Queens High School for the Sciences at York (Q). The difference in preferences towards Stuyvesant and Brooklyn Tech seems to be more nuanced. The complete map of school districts in New York City can and the map of specialized high schools can be found in Appendix \ref{sec:app:sch-districts} and \ref*{sec:app:shp-map}. 

Lastly, we would like to point out some concerns that are not directly observable from our data.~\citet{aygun2020designing} noted that for admissions to Indian Institutes of Technology (IIT), there is often social stigma associated with reserved seats and thus, many students prefer to not be admitted via reserved seats. We also note that NYC DOE defines disadvantaged students based on their social economic status instead of a caste system as in the case of IIT admission. Hence, the severity of the social stigma associated with reserved seats might differ between these two markets. 

In sum, we believe more study is needed to understand students' preference structure over reserved and general seats for the NYC SHS market. Moreover, as a future direction, it would be interesting to design and study mechanisms which incorporate students' preferences over general and reserved seats at all schools, possibly in orders that are not consistent with those interpreted by the mechanisms. 

\begin{figure}[ht]
    \centering
    \includegraphics[width=.72\textwidth]{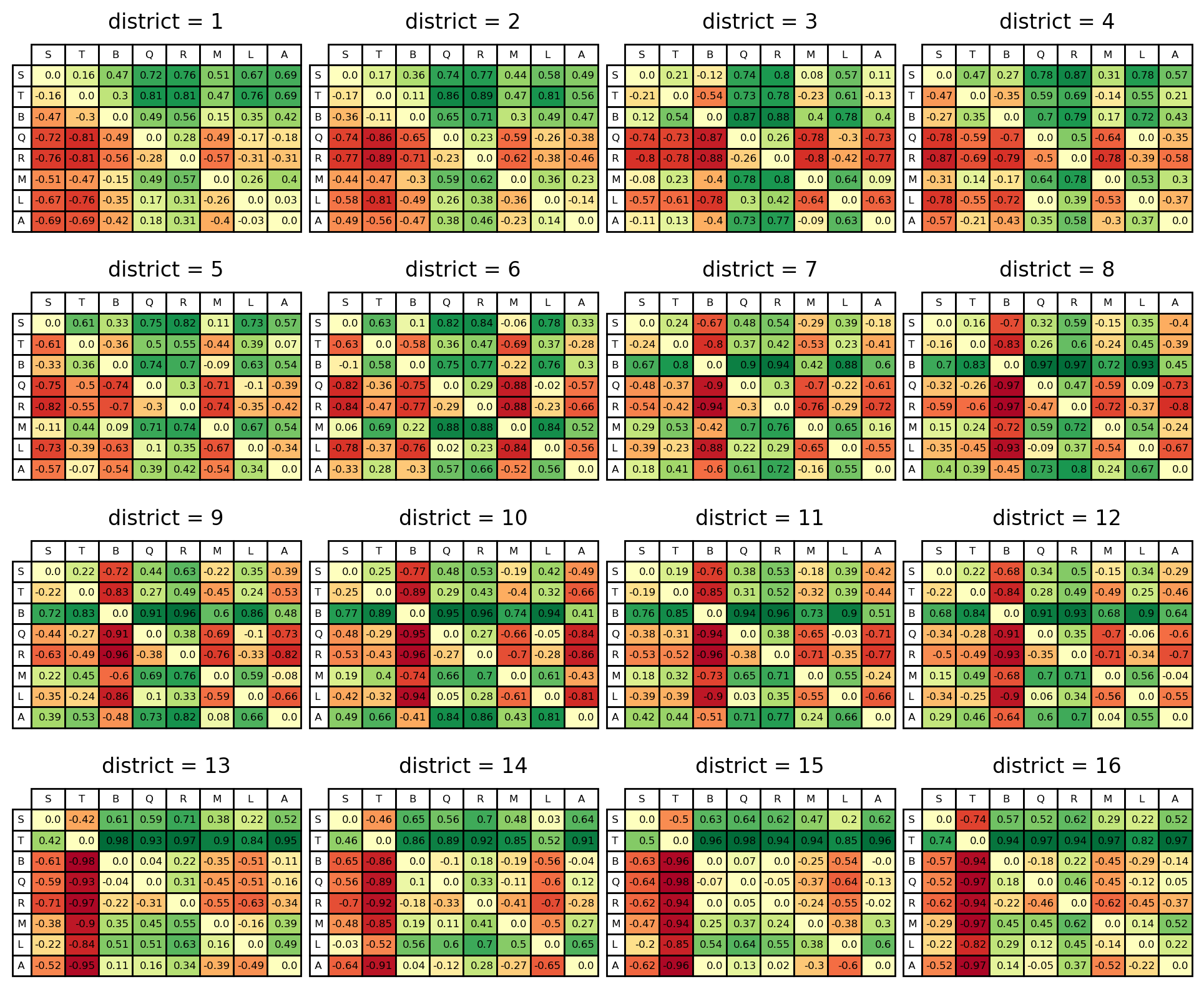}
    \caption{\raggedright These tables are the same as those in Figure \ref{fig:pwc_sch}, but for districts 1 -- 16. }
    \label{fig:pwc_sch_all_1}
\end{figure}
\text{~}
\vspace{-.2cm}
\text{~}
\begin{figure}[h]
    \centering
    \includegraphics[width=.72\textwidth]{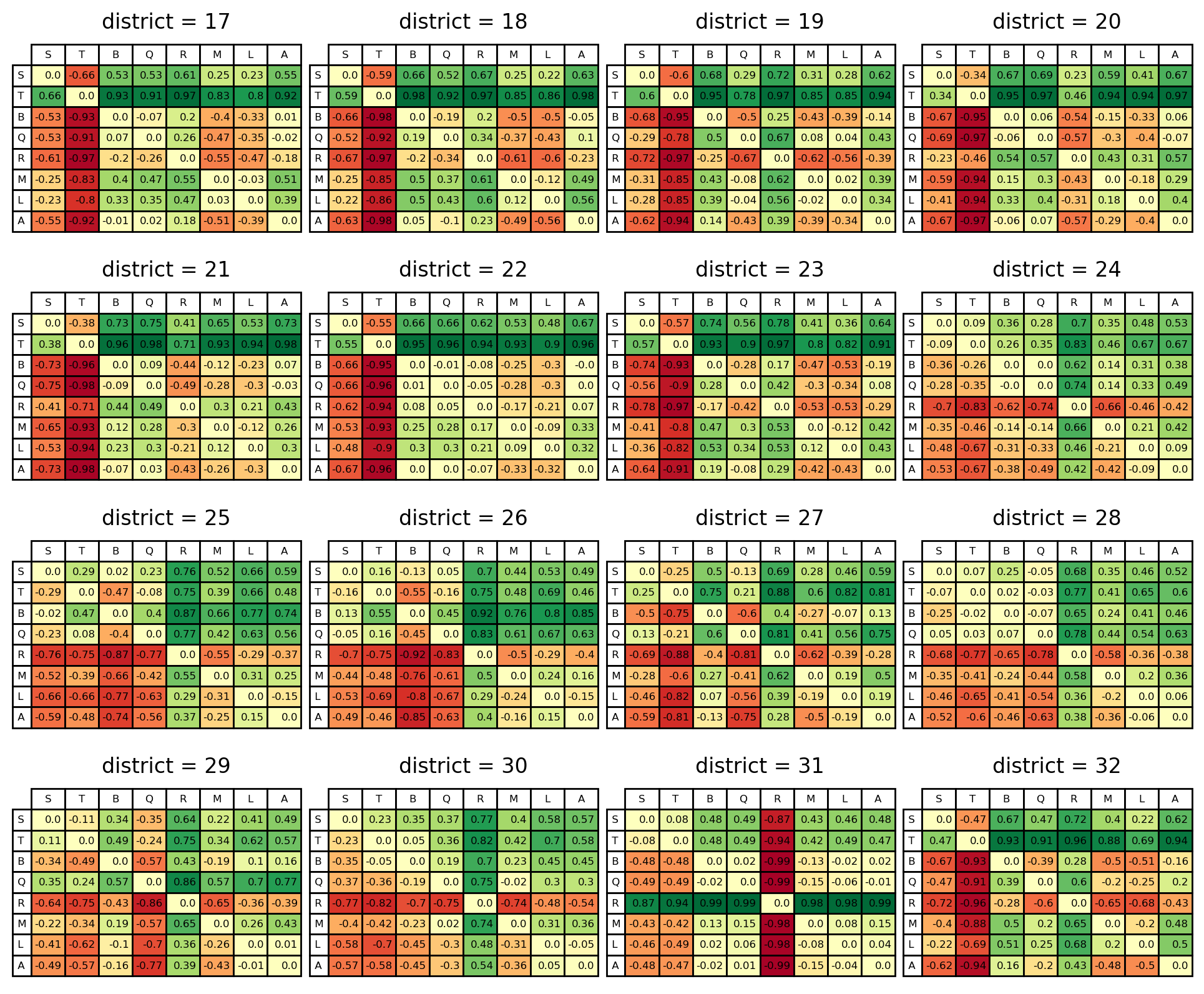}
    \caption{\raggedright These tables are the same as those in Figure \ref{fig:pwc_sch}, but for districts 17 -- 32. }
    \label{fig:pwc_sch_all_2}
\end{figure}
\text{~}

\text{~}

\newpage
\section{Map of NYC School Districts} \label{sec:app:sch-districts}
\begin{figure}[ht]
    \centering
    \includegraphics[width=.45\textwidth]{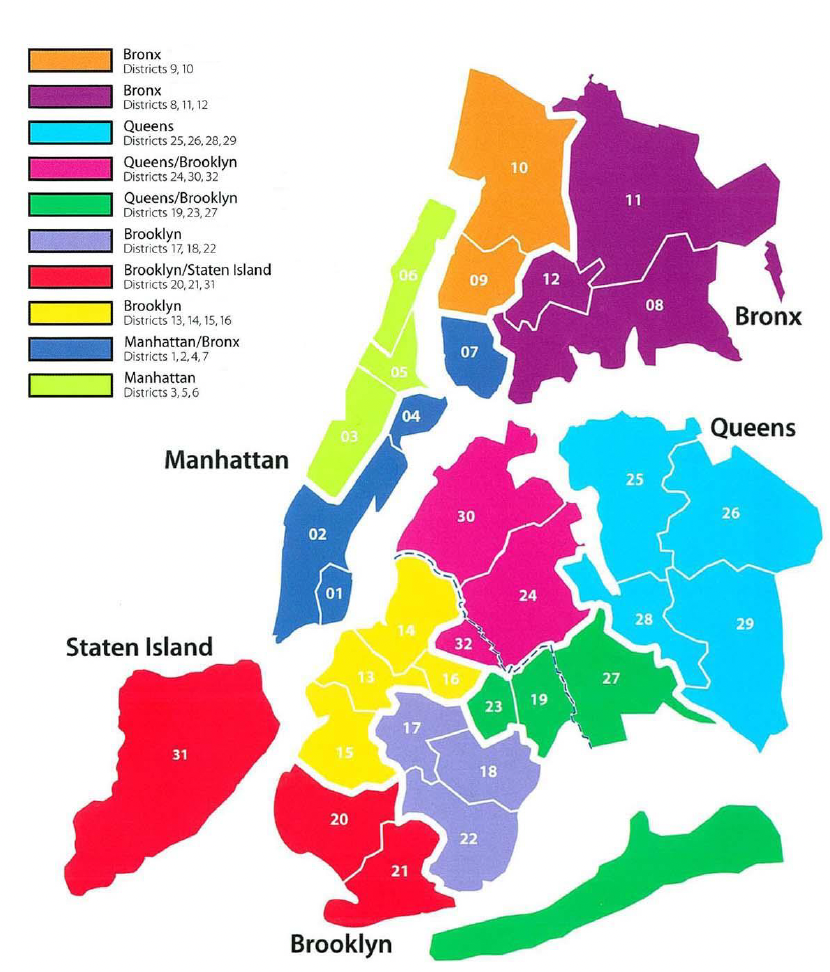}
    \caption{\raggedright Map of school districts in New York City, compiled by NYC DOE and available online at \url{https://video.eschoolsolutions.com/udocs/DistrictMap.pdf}} \label{fig:district}
\end{figure}
\text{~}

\section{Map of NYC Specialized High Schools} \label{sec:app:shp-map}
\begin{figure}[ht]
    \centering
    \includegraphics[width=.55\textwidth]{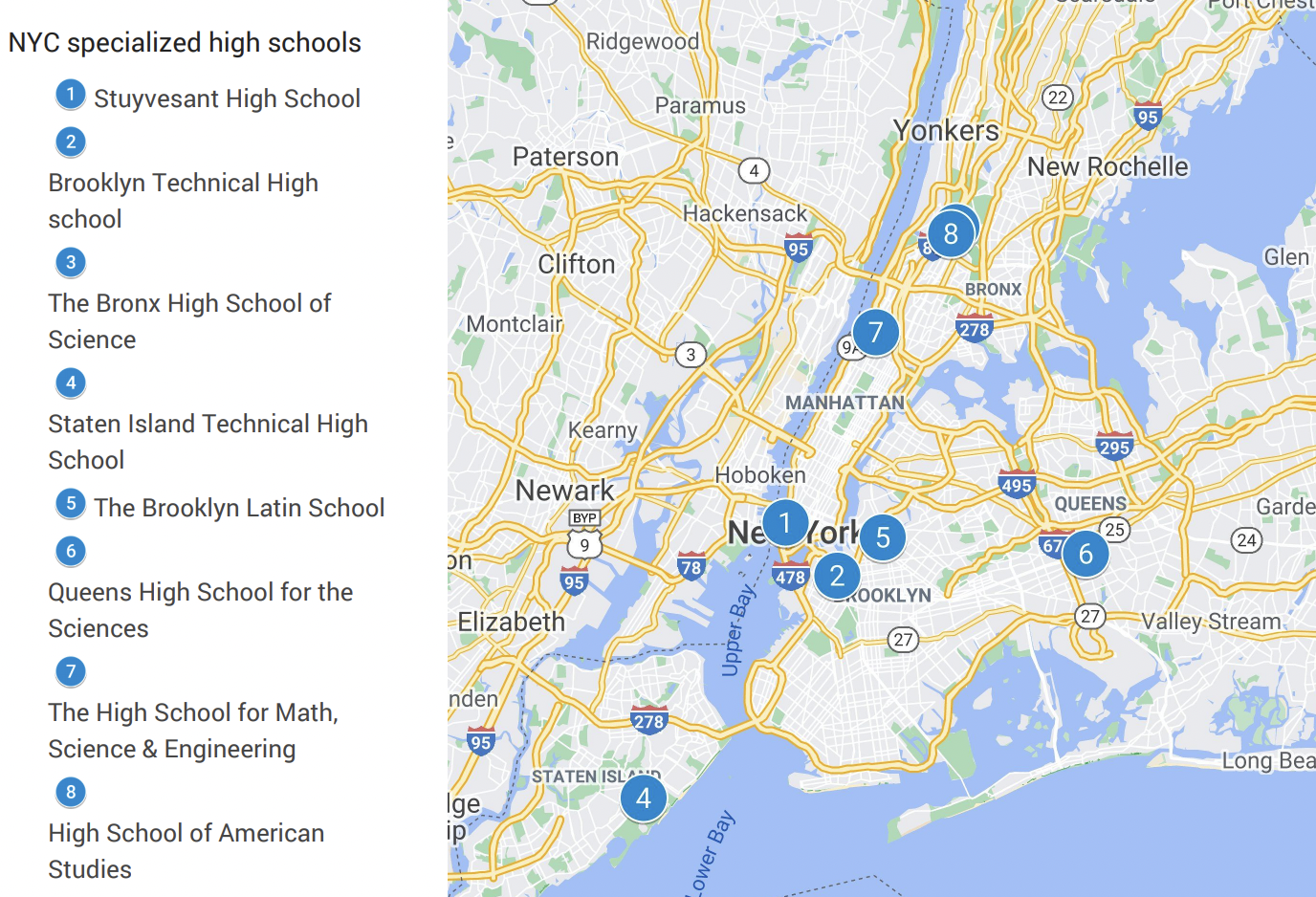}
    \caption{\raggedright Map of specialized high schools in New York City. In Bronx, the two schools numbered by 3 and 8 are overlapping on the map. The map is generated by \emph{Google My Maps}.}
    \label{fig:sph}
\end{figure}
\text{~}

\end{document}